\documentclass[11pt,draftcls,onecolumn]{IEEEtran}
\usepackage{amsbsy}
\usepackage{floatflt} 

\usepackage{amsmath}
\usepackage{amssymb}
\usepackage{times}
\usepackage{graphicx}
\usepackage{xspace}
\usepackage{paralist} 
\usepackage{setspace} 
\usepackage{xypic}
\xyoption{curve}
\usepackage{latexsym}
\usepackage{theorem}
\usepackage{ifthen}
\usepackage{subfigure}
\usepackage{turnstile}

\usepackage{booktabs}
{\theoremheaderfont{\it} \theorembodyfont{\rmfamily}
\newtheorem{theorem}{Theorem}
\newtheorem{lemma}[theorem]{Lemma}
\newtheorem{definition}[theorem]{Definition}

\newtheorem{proposition}[theorem]{Proposition}
\newtheorem{remark}[theorem]{Remark}

}

\renewcommand{\baselinestretch}{1.3}

\begin{document}
\title{Gossip and Distributed Kalman Filtering: Weak Consensus under Weak Detectability}
\author{Soummya Kar$^\dagger$\thanks{$^\dagger$The authors are with the Department of Electrical and Computer Engineering, Carnegie Mellon University, 5000 Forbes Ave, Pittsburgh, PA 15213 (Tel: (412)268-6341, Fax:(412)268-3890). {\tt\small\{soummyak, moura\}@ece.cmu.edu}. This work was partially supported by NSF under grants
\#~ECS-0225449 and~\#~CNS-0428404, and by ONR under grant \#~MURI-N000140710747.}
~and~Jos\'e~M.~F.~Moura$^\dagger$
\\
}

\maketitle \thispagestyle{empty} \maketitle

\begin{abstract}The paper presents the gossip interactive Kalman filter (GIKF) for distributed Kalman filtering for networked systems and sensor networks, where inter-sensor communication and observations occur at the same time-scale. The communication among sensors is random; each sensor occasionally exchanges its filtering state information with a neighbor depending on the availability of the appropriate network link. We show that under a weak distributed detectability condition:
\begin{inparaenum}[1)]
\item the GIKF error process remains stochastically bounded, irrespective of the instability properties of the random process dynamics; and
\item the network achieves \emph{weak consensus}, i.e., the conditional estimation error covariance at a (uniformly) randomly selected sensor converges in distribution to a unique invariant measure on the space of positive semi-definite matrices (independent of the initial state.)
\end{inparaenum}
To prove these results, we interpret the filtered states (estimates and error covariances) at each node in the GIKF as stochastic particles with local interactions. We analyze the asymptotic properties  of the error process by studying as a random dynamical system the associated switched (random) Riccati equation, the switching being dictated by a non-stationary Markov chain on the network graph.
\end{abstract}
\newpage
\section{Introduction}
\label{introduction}
\subsection{Background and Motivation}
\label{backmot}
This paper presents the Gossip Interactive Kalman Filtering (GIKF). GIKF is a linear distributed estimator that filters noisy observations of a random process measured by a sparsely connected sensor network. Each sensor observes only a portion of the process, such that, acting alone, no sensor can resolve the signal.  GIKF is fundamentally different from other distributed implementations of the Kalman filter (\cite{raowhyte-91,olfati-saber:2005,Khan-Moura,Schenato-Kalman})
 that employ some form of linear consensus on the sensor observations or estimates; in contrast, GIKF involves communication and observation sampling at the same time scale. GIKF runs at each sensor a local copy of the Kalman filter and achieves collaboration through occasional asynchronous state swaps between sensors at random time instants. At the random times when a sensor communicates with one of its randomly selected neighbors, the sensor swaps its previous state (its local Kalman filter state estimate and conditional error covariance) with the state of its neighbor, before processing the current observation. In other words, when communication is established, a sensor updates the state it receives from its neighbor with its present observation; otherwise, it updates its own previous state. Such collaboration or information exchange through state swapping is asynchronous over the network and occurs occasionally, dictated by the random network topology. Indeed, due to inherent environmental randomness, the underlying medium access control (MAC) protocol is randomized and often not known at the local sensor level. We assume that the sensor network uses a generic random communication protocol, see Section~\ref{prob_form}, that subsumes the widely used gossiping protocol for real time embedded architectures, \cite{Boyd-GossipInfTheory}, and the graph matching based communication protocols for internet architectures, \cite{Mckeown}.

The paper establishes GIKF and studies its error properties. We define a weak distributed detectability condition\footnote{This condition is required even by a centralized estimator (having access to all sensor observations over all time) to yield an estimate with bounded error (for unstable systems.)} under which we show:
\begin{inparaenum}[1)]
\item the GIKF error process remains stochastically bounded, irrespective of the instability properties of the random process dynamics; and
\item the network achieves weak consensus, i.e., the conditional estimation error covariance at a (uniformly) randomly selected sensor converges in distribution to a unique invariant measure on the space of positive semi-definite matrices (independent of the initial state.)
\end{inparaenum}
To prove these results, we interpret the filtered states at each node in the GIKF as stochastic particles with local interactions and analyze the asymptotic properties  of the error process by studying as a random dynamical system the switched (random) Riccati equation, the switching being dictated by a non-stationary Markov chain on the network graph.

To study the information flows in the GIKF, we interpret the filtering states at each node as stochastic particles with controlled interactions. To prove the stochastic boundedness of the error process and the network weak consensus, we focus on these traveling states, which we refer to as tokens or particles, and not on the sequence of conditional error covariances at each sensor, which is not Markov. This particle point of view is reminiscent of the approach taken in fluid dynamics of studying the transport of a particle as it travels in the fluid (Lagrangian coordinates) rather than studying the transport at a fixed coordinate in space (Eulerian coordinates), \cite{batchelor-73}. We show that the sequence of traveling states or particles evolves according to a switched system of random Riccati operators, the switching being dictated by a non-stationary Markov chain on the graph. A key contribution is the analysis of the resulting random Riccati equation (RRE). In this context, we note that the RRE arises in the literature in several practical filtering and control formulations with non-classical information. Prior work (\cite{ckvs89,wg99,ms03,Bruno,Liu:04,Gupta:05,Xu:05,Minyi:07,kp-fb:07j,Craig:07,xx07}) mostly address qualitative properties of the RRE in terms of moment stability, whereas recent approaches focus on understanding the limit behavior in terms of weak convergence (\cite{Riccati-weakconv,Riccati-moddev,censi09,vh08}, see also~\cite{Bougerol}). In this paper, we utilize a random dynamical systems formulation of the RRE; however, in contrast with our work in~\cite{Riccati-weakconv,Riccati-moddev}, the switching sequence is no longer stationary. Several approximation arguments of independent interest are developed to tackle this non-stationary behavior and to establish the asymptotic distributional properties of the RRE.

To summarize, the paper addresses two fundamental concerns in collaborative estimation in random environments. It introduces \emph{distributed observability} for linear dynamical estimation and addresses the question of \emph{minimal} observation pattern (i.e., what should be the minimal number of sensors and what should they observe,) so that there exists a \emph{successful} filtering scheme. The weak detectability condition (introduced in Section~\ref{setup}) resolves this question through the existence of a full rank network Grammian. We show that satisfaction of the weak detectability condition leads to stochastic boundedness of the conditional filtering error at each sensor, irrespective of observability of individual sensors. The second concern addressed in the paper is that of robust information flow, which seeks to address the minimal communication required to maintain consistent (asymptotically) information dissemination in the network. The weak connectedness assumption formulated in Section~\ref{setup} quantifies the rate of information flow (in random communication environments) as the mixing time of a particle undergoing a random walk in the network with appropriate statistics. The positive recurrence of this Markov chain translates to information dissemination at a sufficient rate to cope with the (possible) instability in signal dynamics and leads to weak consensus of the filtering errors. The notion of weak consensus introduced in the paper is the best form of \emph{consensus} possible in such a setup because, as opposed to familiar scenarios (average computation/static parameter estimation,) in a dynamic situation it is not possible to accomplish almost sure (pathwise) consensus of the estimate or error processes. On the contrary, the weak consensus we establish shows that the error processes at different sensors converge in distribution to the same invariant measure. We do not characterize here this invariant measure as a function of the communication and observation policies; instead, we resolve the minimal conditions for the existence of such an invariant measure and hence conditions for the stability of the filtering error processes.

We briefly summarize the organization of the rest of the paper. Subsection~\ref{notprel} sets up notation and background material
to be used in the paper. Section~\ref{prob_form} sets-up the problem and  introduces the GIKF algorithm together with the observability and connectivity assumptions in Subsection~\ref{setup}. An interactive particle interpretation and important preliminary results are in Subsection~\ref{inter_part}.  The main results regarding the asymptotic properties of the GIKF are stated (without proof) and interpreted in Section~\ref{main_res}. To prove these results, we provide first in Section~\ref{aux} a random dynamic system~(RDS) formulation of the switching iterates of the random Riccati equation arising in the GIKF. Appendix~\ref{RDS} recalls facts and results on random dynamical systems~(RDS) needed in this Section. The main results of the paper are proved in Section~\ref{proof_main_res}. Two technical Lemmas are proven in Appendix~\ref{RDS_123}. Finally, Section~\ref{conclusion} concludes the paper.

\subsection{Notation and Preliminaries}
\label{notprel}
Let $\mathbb{R}$ be the reals; $\mathbb{R}^{M}$,
the $M$-dimensional Euclidean space; $\mathbb{T}$, the integers; $\mathbb{T}_{+}$, the non-negative integers; $\mathbb{N}$, the natural numbers; and $\mathcal{X}$, a generic space.
For  $B\subset \mathcal{X}$,  $\mathbb{I}_{B}:\mathcal{X}\longmapsto\{0,1\}$ is the indicator
function, i.e., $1$ when its argument is in~$B$ and zero otherwise; and $\mbox{id}_{\mathcal{X}}$ is the identity function on
$\mathcal{X}$.


\textbf{Cones in partially ordered Banach spaces.}
 We summarize facts and definitions on the structure of
cones in partially ordered Banach spaces. Let $V$ be a Banach
space (over the field of the reals) with a closed (w.r.t. the
Banach space norm) convex cone $V_{+}$ and assume $V_{+}\cap
(-V_{+})=\{0\}$. The cone $V_{+}$ induces a partial order in $V$,
namely, for $X,Y\in V$, we write $X\preceq Y$, if $Y-X\in V_{+}$.
In case $X\preceq Y$ and $X\neq Y$, we write $X\prec Y$. The cone
$V_{+}$ is called solid, if it has a non-empty interior
$\mbox{int}\,V_{+}$; in that case, $V_{+}$ defines a strong
ordering in $V$, and we write $X\ll Y$, if
$Y-X\in\mbox{int}\,V_{+}$. The cone $V_{+}$ is normal if the norm
$\|\cdot\|$ of $V$ is semi-monotone, i.e., $\exists\, c>0$, s.t.
$0\preceq X\preceq Y \Rightarrow \|X\|\leq c\|Y\|$. There are
various equivalent characterizations of normality, of which we
note that the normality of $V_{+}$ ensures that the topology in
$V$ induced by the Banach space norm is compatible with the
ordering induced by $V_{+}$, in the sense that any norm-bounded
set $B\subset V$ is contained in a conic interval of the form
$[X,Y]$, where $X,Y\in V$. Finally, a cone is said to be
minihedral, if every order-bounded (both upper and lower bounded)
finite set $B\subset V$ has a supremum (here bounds are w.r.t. the
partial order.)

We focus on the separable Banach
space of symmetric $n\times n$ matrices, $\mathbb{S}^{n}$,
equipped with the induced 2-norm. The subset $\mathbb{S}^{N}_{+}$
of positive semidefinite matrices is a closed, convex, solid,
normal, minihedral cone in $\mathbb{S}^{n}$, with non-empty
interior $\mathbb{S}^{N}_{++}$, the set of positive definite
matrices. The conventions above denote the
partial and strong ordering in $\mathbb{S}^{n}$ induced by
$\mathbb{S}_{+}^{N}$.

\textbf{Probability measures on metric spaces:} Let: $(\mathcal{X},d_{\mathcal{X}})$ a
complete separable metric space $\mathcal{X}$ with
metric~$d_{\mathcal{X}}$;  $\mathbb{B}(\mathcal{X})$ its Borel algebra;
 $B(\mathcal{X})$ the Banach space of real-valued bounded
functions on $\mathcal{X}$, equipped with the sup-norm, i.e.,
$f\in B(\mathcal{X}), \:\|f\|=\sup_{x\in\mathcal{X}}|f(x)|$; and $C_{b}(\mathcal{X})$  the subspace of $B(\mathcal{X})$ of continuous
functions. For $x\in\mathcal{X}$, the open ball of radius $\varepsilon>0$ centered at $x$ is denoted by $B_{\varepsilon}(x)$, i.e., $B_{\varepsilon}(x)=\left\{y\in\mathcal{X}~|~d_{\mathcal{X}}(y,x)<\varepsilon\right\}$.
For any set $\Gamma\subset\mathcal{X}$, the open $\varepsilon$-neighborhood of $\Gamma$ is given by $\Gamma_{\varepsilon}=\left\{y\in\mathcal{X}~|~\inf_{x\in\Gamma}d_{\mathcal{X}}(y,x)<\varepsilon\right\}$.
It can be shown that $\Gamma_{\varepsilon}$ is an open set.

Let $\mathcal{P}(\mathcal{X})$ be the
set of probability measures on $\mathcal{X}$.
A sequence $\{\mu_{t}\}_{t\in\mathbb{T}_{+}}$ of probability measures in
$\mathcal{P}(\mathcal{X})$ converges weakly to $\mu\in \mathcal{P}(\mathcal{X})$ if $\lim_{t\rightarrow\infty}<f,\mu_{t}>\,=\,<f,\mu>,~~\forall~f\in
C_{b}(\mathcal{X})$.
By Portmanteau's theorem, the above is equivalent to any one of
the following:
\begin{itemize}
\item[]\mbox[i] For all closed $F\in\mathbb{B}(\mathcal{X})$\hspace{1cm}
$\limsup_{t\rightarrow\infty}\mathbb{\mu}_{t}(F)\leq\mu(F)$
\item[]\mbox[ii] For all open $O\in\mathbb{B}(\mathcal{X})$ \hspace{1cm}
$\liminf_{t\rightarrow\infty}\mu_{t}(O)\geq\mu(O)$
\end{itemize}
Weak convergence is denoted by $\mu_{t}\Longrightarrow\mu$ and is
also referred to as convergence in distribution. The weak topology
on $\mathcal{P}(\mathcal{X})$ generated by weak convergence can be
metrized. In particular, e.g., \cite{Jacod-Shiryaev}, one has the
Prohorov metric $d_{p}$ on $\mathcal{P}(\mathcal{X})$, such that
the metric space $\left(\mathcal{P}(\mathcal{X}),d_{p}\right)$ is
complete, separable, and a sequence
$\{\mu_{t}\}_{t\in\mathbb{T}_{+}}$ in $\mathcal{P}(\mathcal{X})$
converges weakly to $\mu$ in $\mathcal{P}(\mathcal{X})$ \emph{iff}
$\lim_{t\rightarrow\infty}d_{p}(\mu_{t},\mu)=0$.
The distance between two probability measures $\mathbb{\mu}_{1},\mathbb{\mu}_{2}$ in $\mathcal{P}(\mathcal{X})$ is computed as:
\begin{equation}
\label{WC23}
d_{P}\left(\mathbb{\mu}_{1},\mathbb{\mu}_{2}\right)=\inf\left\{\varepsilon>0~|~\mathbb{\mu}_{1}(\mathcal{F})\leq \mathbb{\mu}_{2}(\mathcal{F}_{\varepsilon})+\varepsilon,~~~\forall~\mbox{closed set}~\mathcal{F}\right\}
\end{equation}

\section{Gossip Interactive Kalman Filter (GIKF)}
\label{prob_form}

\subsection{Problem setup}
\label{setup}

\textbf{Signal/Observation Model}
We consider a discrete-time linear dynamical system observed by a network of $N$ sensors. The signal model is:
\begin{equation}
\label{sys_model}
\mathbf{x}_{t+1}=\mathcal{F}\mathbf{x}_{t}+\mathbf{w}_{t}
\end{equation}
where $\mathbf{x}_{t}\in\mathbb{R}^{M}$ is the signal (state)
vector with initial state $\mathbf{x}_{0}$ distributed as a
zero mean Gaussian vector with covariance $\widehat{P}_{0}$ and
the system noise $\{\mathbf{w}_{t}\}$ is an uncorrelated zero mean
Gaussian sequence independent of $\mathbf{x}_{0}$ with covariance
$\mathcal{Q}$. The observation at the $n$-th sensor
$\mathbf{y}^{n}_{t}\in\mathbb{R}^{m_{n}}$ at time $t$ is:
\begin{equation}
\label{obs_n}
\mathbf{y}^{n}_{t}=\mathcal{C}_{n}\mathbf{x}_{t}+\mathbf{v}^{n}_{t}
\end{equation}
where $\mathcal{C}_{n}\in\mathbb{R}^{m_{n}\times M}$ and
$\{\mathbf{v}^{n}_{t}\}$ is an uncorrelated zero mean Gaussian
observation noise sequence with covariance $\mathcal{R}_{n}\gg
\mathbf{0}$. Also, the noise sequences at different sensors are
independent of each other, the system noise process and the
initial system state. Because of the limited capability of the
sensors, typically the dimension of $\mathbf{y}^{n}_{t}$ is much
smaller than that of the signal process and the observation
process at each sensor is not sufficient to make the pair
$\{\mathbf{x}_{t},\mathbf{y}^{n}_{t}\}$ observable\footnote{It is possible that some of the sensors have no observation
capabilities, i.e., the corresponding $C_{n}$ is a zero matrix.
Thus the formulation easily carries over to networks of
heterogeneous agents, consisting of `sensors' which actually sense
the field of interest and actuators, which implement local control
actions based on the estimated field.}. We envision a totally
distributed application where a reliable estimate of the signal
process is required at each sensor.\footnote{The term sensor
network here refers to a network of \emph{agents} (possibly
distributed over a geographical region) with varied
functionalities. For example, some agents may be physical sensors
while others may be remote actuators, in which case, the
corresponding observation matrix $C_{n}$ is identically zero. In
this paper, we use the term sensor to denote a generic network
agent.} The sensors achieve collaboration with each other by means
of occasional communication with their \emph{neighbors}, whereby
they exchange their filtering states (to be defined precisely.) We
assume that time is slotted and inter-sensor communication and
sensing (observation) take place at the same time-scale.

\textbf{Communication Model}
Communication among sensors is constrained by several factors including proximity, transmit power, and receiving capabilities. We model the underlying communication structure of the network in terms of an undirected graph $(V,\mathcal{E})$ where $V$ denotes the set of $N$ sensors and $\mathcal{E}$ is the set of edges or allowable communication links between the sensors. The notation $n\sim l$ indicates that sensors $n$ and $l$ can communicate, i.e., $\mathcal{E}$ contains the undirected edge $(n,l)$. The graph can be represented in terms of its $N\times N$ symmetric adjacency matrix $\mathcal{A}$:
\begin{equation}
\label{def_mathcalA}
\mathcal{A}_{nl}=\left\{ \begin{array}{ll}
                    1 & \mbox{if $(n,l)\in \mathcal{E}$}\\
                    0 & \mbox{otherwise}
                   \end{array}
          \right.
\end{equation}
We assume that the diagonal elements of $\mathcal{A}$ are
identically 1, indicating that a sensor $n$ can always communicate
to itself. Note, that $\mathcal{E}$ is the maximal allowable set
of links in the network at any time, however, at a particular
instant, each sensor may choose to communicate only to a fraction
of its neighbors. The exact communication protocol is not so
important for the analysis, as long as some \emph{weak}
connectivity assumptions are satisfied. For definiteness, we
assume the following generic communication model, which subsumes
the widely used gossiping protocol for real time embedded
architectures (\cite{Boyd-GossipInfTheory}) and the graph matching
based communication protocols for internet architectures
(\cite{Mckeown}.) We make this precise in the following, which we generalize later. Define the set $\mathcal{M}$
of symmetric 0-1 $N\times N$ matrices:
\begin{equation}
\label{def_mathcalM}\mathcal{M}=\left\{A~\left|~\mathbf{1}^{T}A=\mathbf{1}^{T},~~A\mathbf{1}=\mathbf{1},~~A\leq \mathcal{E}\right.\right\}
\end{equation}
In other words, $\mathcal{M}$ is the set of adjacency matrices,
such that, every node is incident to exactly one edge (including
self edges) and allowable edges are only those included in
$E$.\footnote{The set $\mathcal{M}$ is always non-empty, in
particular the $N\times N$ identity matrix $I_{N}\in\mathcal{M}$.}
Let $\mathcal{D}$ be a probability distribution on the space
$\mathcal{M}$. The sequence of time-varying adjacency matrices,
$\{A(t)\}_{t\in\mathbb{N}}$, governing the inter-sensor
communication, is then an i.i.d. sequence in $\mathcal{M}$ with
distribution $\mathcal{D}$ and independent of the signal and
observation processes.\footnote{For convenience of presentation,
we assume that $A(0)=I_{N}$, although communication starts at slot
$t=1$.} We make the following assumption of connectivity on the
average:

\textbf{Assumption C.1}: Define the symmetric stochastic matrix
\label{assumptionc.1}
$\overline{A}$ as
\begin{equation}
\label{def_barA}\overline{A}=\mathbb{E}\left[A(t)\right]=\int_{\mathcal{M}}Ad\mathcal{D}(A)
\end{equation}
The matrix $\overline{A}$ is assumed to be irreducible and aperiodic.

\begin{remark}
\label{rem1} The stochasticity of $\overline{A}$ is inherited from
that of the elements of $\mathcal{M}$. We are not concerned
with the properties of the distribution $\mathcal{D}$ as long as
the \emph{weak} connectivity assumption above is satisfied. The
issue of $\overline{A}$ being irreducible depends both on the set
of allowable edges $\mathcal{E}$ and the distribution
$\mathcal{D}$. We do not pursue that question in detail here.
However, to show the applicability of Assumption~\textbf{C.1} and
justify the notion of weak connectivity, we note that such a
distribution $\mathcal{D}$ always exists if the graph
$(V,\mathcal{E})$ is connected. We give a Markov chain
interpretation of the mean adjacency matrix $\overline{A}$, which
will be helpful for the analysis to follow. The matrix
$\overline{A}$ can be associated to the transition kernel of a
time-homogeneous Markov chain on the state space $V$. Since the
state space $V$ is finite, the irreducibility of $\overline{A}$
suggests that the resulting Markov chain is positive recurrent.
Due to symmetricity, the Markov chain is reversible with unique
invariant distribution $\pi$ on $V$, where $\pi$ is the discrete
uniform distribution on $V$.
\end{remark}

\textbf{Observability Conditions: Weak Detectability}
Successful filtering even in the centralized setting (assuming
all the sensors can forward their observations at all time to a
fusion center) requires some form of detectability and
stabilizability. In the present distributed setting we impose the
following weak assumptions on the signal/observation model:

\textbf{Stabilizability: Assumption S.1} The pair
\label{assumptions.1}
$(\mathcal{F},\mathcal{Q}^{1/2})$ is stabilizable. The
non-degeneracy of $\mathcal{Q}$ ensures this.

For distributed detectability, we assume the following:

%

\textbf{Weak Detectability: Assumption D.1} There exists  a
\label{assumptiond.2}
walk\footnote{A walk in this context is defined w.r.t.~the graph
induced by the non-zero entries of the matrix $\overline{A}$.} of
length $\ell\geq 1$, $\left(n_{1},n_{2},\cdots, n_{\ell}\right)$, covering
the $N$ nodes, such that, the matrix
$\sum_{i=1}^{\ell}\left(\mathcal{F}^{i-1}\right)^{T}\mathcal{C}_{n_{i}}^{T}\mathcal{C}_{n_{i}}\mathcal{F}^{i-1}$
is invertible.
\begin{remark}
\label{rem2} Note, as permitted by the general definition of a
walk, the sequence $\left(n_{1},n_{2},\cdots, n_{\ell}\right)$ may consist of
repeated vertices and, in particular, self-loops (if permitted by
$\overline{A}$.)
\end{remark}
\begin{remark}
When
$\mathcal{F}$ is invertible, \textbf{D.1} may be replaced by the full rank of
\begin{equation}
\label{def_G}
\mathcal{G}=\sum_{n=1}^{N}\mathcal{C}_{n}^{T}\mathcal{C}_{n}
\end{equation}
Indeed, by the irreducibility of $\overline{A}$
(equivalently, by the connectivity of the graph induced by
$\overline{A}$,) we can find a walk $\left(n_{1},n_{2},\cdots, n_{\ell}\right)$ of length $\ell\geq N$, which covers the network, i.e.,
visits each node at least once. Hence, if $\mathcal{F}$ is
invertible and~(\ref{def_G}) holds, it follows that the
matrix
$\sum_{i=1}^{\ell}\left(\mathcal{F}^{i-1}\right)^{T}\mathcal{C}_{n_{i}}^{T}\mathcal{C}_{n_{i}}\mathcal{F}^{i-1}$
corresponding to this walk is invertible leading to
Assumption~\textbf{D.1}.
\end{remark}
\begin{remark}
From the positive definiteness of the measurement noise
matrices $\mathcal{R}_{n}$, it follows that under~\textbf{D.1},
the matrix
$\sum_{i=1}^{\ell}\left(\mathcal{F}^{i-1}\right)^{T}\mathcal{C}_{n_{i}}^{T}\mathcal{R}_{n_{i}}^{-1}
\mathcal{C}_{n_{i}}\mathcal{F}^{i-1}$
is invertible.
\end{remark}
\begin{remark}
Assumption~\textbf{D.1} is minimal, in the
sense, that, even in a centralized setting (a center has access to
all the sensor observations over all time,) it is required to
ensure detectability for arbitrary choice of the matrix $F$
governing the signal dynamics. This justifies the nomenclature
weak detectability.
\end{remark}

\textbf{Algorithm GIKF}
We now present the algorithm GIKF (gossip based interacting Kalman
filter) for distributed estimation of the signal process
$\mathbf{x}_{t}$ over time. We start by introducing notation. Let the filter at sensor $n$ be initialized with the pair
$\left(\widehat{\mathbf{x}}_{0|-1},\widehat{P}_{0}\right)$, where
$\widehat{\mathbf{x}}_{0|-1}$ denotes the prior estimate of
$\mathbf{x}_{0}$ (with no observation information) and
$\widehat{P}_{0}$ the corresponding error covariance. Also, by
$(\widehat{\mathbf{x}}_{t|t-1}^{n},\widehat{P}_{t}^{n})$ denote
the estimate at sensor $n$ of $\mathbf{x}_{t}$ based on
information\footnote{The information at sensor $n$ till (and
including) time $t$ corresponds to the sequence of observations
$\{\mathbf{y}^{n}_{s}\}_{0\leq s\leq t}$ obtained at the sensor
and the information received by data exchange between its
neighboring senors.} till time $t-1$ and the corresponding
conditional error covariance, respectively. The pair
$\left(\widehat{\mathbf{x}}_{t|t-1}^{n},\widehat{P}_{t}^{n}\right)$ is also
referred to as the state of sensor $n$ at time $t-1$. To define
the estimate update rule for the GIKF, let $\rightarrow(n,t)$ be
the neighbor of sensor $n$ at time $t$ w.r.t.~the adjacency matrix \footnote{Note that $n(t)$ is unambiguously defined as $A(t)$ is a matching matrix, and also by symmetry we have $\rightarrow(\rightarrow(n,t),t)=n$. It is possible that
$\rightarrow(n,t)=n$, in which case the graph corresponding to $A(t)$ has a self-loop at node $n$.} $A(t)$. We assume that all inter-sensor communication for time $t$ occurs at the beginning of the slot, whereby communicating sensors swap their previous states, i.e., if at time $t$, $\rightarrow(n,t)=l$, sensor $n$ replaces its previous state $\left(\widehat{\mathbf{x}}_{t|t-1}^{n},\widehat{P}_{t}^{n}\right)$ by $\left(\widehat{\mathbf{x}}_{t|t-1}^{l},\widehat{P}_{t}^{l}\right)$ and sensor $l$ replaces its previous state $\left(\widehat{\mathbf{x}}_{t|t-1}^{l},\widehat{P}_{t}^{l}\right)$ by $\left(\widehat{\mathbf{x}}_{t|t-1}^{n},\widehat{P}_{t}^{n}\right)$. The estimate update at sensor $n$ at the end of the slot (after the communication and observation tasks have been completed) is:
{\small
\begin{eqnarray}
\label{est_up}
\widehat{\mathbf{x}}_{t+1|t}^{n}&=&\mathbb{E}\left[\mathbf{x}_{t+1}~\left|~\widehat{\mathbf{x}}_{t|t-1}^{\rightarrow(n,t)},
\widehat{P}_{t}^{\rightarrow(n,t)},\mathbf{y}^{n}_{t}\right.\right]\\
\label{est_up1}
\widehat{P}_{t+1}^{n}&=&\mathbb{E}\left[\left(\mathbf{x}_{t+1}-\widehat{\mathbf{x}}_{t+1|t}^{n}\right)
\left(\mathbf{x}_{t+1}-\widehat{\mathbf{x}}_{t+1|t}^{n}\right)^{T}~\left|~\widehat{\mathbf{x}}_{t|t-1}^{\rightarrow(n,t)},
\widehat{P}_{t}^{\rightarrow(n,t)},\mathbf{y}^{n}_{t}\right.\right]
\end{eqnarray}
}
Due to conditional Gaussianity, the filtering steps above can be implemented through the time-varying Kalman filter recursions, and it follows that the sequence $\left\{\widehat{P}_{t}^{n}\right\}$ of conditional predicted error covariance matrices at sensor $n$ satisfies the Riccati recursion:
\begin{equation}
\label{est_up2}
\widehat{P}_{t+1}^{n}=\mathcal{F}\widehat{P}_{t}^{\rightarrow(n,t)}\mathcal{F}^{T}+
\mathcal{Q}-\mathcal{F}\widehat{P}_{t}^{\rightarrow(n,t)}\mathcal{C}_{n}^{T}
\left(\mathcal{C}_{n}\widehat{P}_{t}^{\rightarrow(n,t)}\mathcal{C}_{n}^{T}+
\mathcal{R}_{n}\right)^{-1}\mathcal{C}_{n}\widehat{P}_{t}^{\rightarrow(n,t)}\mathcal{F}^{T}
\end{equation}
Note that the sequence $\left\{\widehat{P}_{t}^{n}\right\}$ is random, due to the random neighborhood selection function $\rightarrow(n,t)$. The goal of the paper is to study asymptotic properties of the sequence of random conditional error covariance matrices $\left\{\widehat{P}_{t}^{n}\right\}$ at every sensor $n$ and show in what sense they reach \emph{consensus}, so that, in the limit of large time, every sensor provides an equally good (stable in the sense of estimation error) estimate of the signal process.

\subsection{An Interacting Particle Representation}
\label{inter_part} To compactify the notation in eqn.~(\ref{est_up2}), we define the functions $f_{n}:\mathbb{S}_{+}^{N}\longmapsto\mathbb{S}_{+}^{N}$ for $n=1,\cdots,N$ defining the respective Riccati operators\footnote{In case a sensor does not observe, i.e., $C_{n}=0$, then the corresponding Riccati operator $f_{n}$ in eqn.~(\ref{def_Riccati}) reduces to the Lyapunov operator.}:
\begin{equation}
\label{def_Riccati}
f_{n}(X)=\mathcal{F}X\mathcal{F}^{T}+\mathcal{Q}-
\mathcal{F}X\mathcal{C}_{n}^{T}\left(\mathcal{C}_{n}X\mathcal{C}_{n}^{T}+
\mathcal{R}_{n}\right)^{-1}\mathcal{C}_{n}X\mathcal{F}^{T}
\end{equation}
Recall the sequence $\{\rightarrow(n,t)\}_{t\in\mathbb{T}_{+}}$ of neighborhoods of sensor $n$. The sequence of conditional error covariance matrices $\left\{P^{n}_{t}\right\}_{t\in\mathbb{T}_{+}}$ at sensor $n$ then evolves according to
\begin{equation}
\label{evolve_P}\widehat{P}_{t+1}^{n}=f_{n}\left(\widehat{P}_{t}^{\rightarrow(n,t)}\right)
\end{equation}
The sequence $\left\{\widehat{P}^{n}_{t}\right\}$ is non-Markov (not even semi-Markov given the random adjacency matrix sequence $\{A(t)\}$,) as $\widehat{P}_{t+1}^{n}$ at time $t$ is a random functional of the conditional error covariance at time $t-1$ of the sensor $\rightarrow(n,t)$, which, in general, is different from sensor $n$. This makes the evolution of the sequence $\left\{\widehat{P}^{n}_{t}\right\}$ difficult to track. To overcome this, we give the following interacting particle interpretation of the conditional error covariance evolution, which naturally leads us to track  semi-Markov sequences of conditional error covariance matrices from which we can completely characterize the evolution of the desired covariance sequences $\left\{\widehat{P}^{n}_{t}\right\}$ for $n=1,\cdots,N$.

To this end, we note that the link formation process given by the sequence $\{A(t)\}$ can be represented in terms of $N$ particles moving on the graph as identical Markov chains. The state of the $n$-th particle is denoted by $p_{n}(t)$, and the sequence $\left\{p_{n}(t)\right\}_{t\in\mathbb{T}_{+}}$ takes values in $[1,\cdots,N]$. The evolution of the $n$-th particle is given as follows:
\begin{equation}
\label{part_n}p_{n}(t)=\rightarrow(p_{n}(t-1),t),~~~p_{n}(0)=n
\end{equation}
Recall the (random) neighborhood selection $\rightarrow(n,t)$. Thus, the $n$-th particle can be viewed as originating from node $n$ at time 0 and then traveling on the graph (possibly changing its location at each time) according to the link formation process $\{A(t)\}$. The following proposition establishes important statistical properties of the sequence $\left\{p_{n}(t)\right\}_{t\in\mathbb{T}_{+}}$:
\begin{proposition}
\label{prop_int}
\item \mbox{[i]} For each $n$, the process $\left\{p_{n}(t)\right\}_{t\in\mathbb{T}_{+}}$ is a Markov chain on $V=[1,\cdots,N]$ with transition probability matrix~$\overline{A}$.
\item \mbox{[ii]} The Markov chain $\left\{p_{n}(t)\right\}_{t\in\mathbb{T}_{+}}$ is ergodic with the uniform distribution on $V$ being the attracting invariant measure.
\end{proposition}
\begin{proof}
For part [i], we note that, by the independence of $\{A(t)\}$, for any $t\in\mathbb{T}_{+}$ and $l_{t},\cdots, l_{0}\in V$,
{\small
\begin{eqnarray}
\label{prop_int1}
\mathbb{P}\left[p_{n}(t)=l_{t}\left|p_{n}(t-1)=l_{t-1},\cdots,p_{n}(1)=l_{1},p_{n}(0)=l_{0}\right.\right] & = & \mathbb{P}\left[p_{n}(t)=l_{t}\left|p_{n}(t-1)=l_{t-1}\right.\right]
\nonumber \\
& = &
\mathbb{P}\left[\rightarrow(l_{t-1},t-1)=l_{t}\right]
\nonumber \\
& = &
\mathbb{P}\left[A_{l_{t-1},l_{t}}(t-1)=1\right]
\nonumber \\
& = & \overline{A}_{l_{t-1},l_{t}}
\end{eqnarray}
}
\noindent
where the last step follows from the fact, that the entries of $A(t-1)$ are binary. This establishes the desired Markovianity of the sequence $\left\{p_{n}(t)\right\}_{t\in\mathbb{T}_{+}}$.

For part [ii], since the state space $V$ is finite, the irreducibility of $\overline{A}$ implies its positive recurrence and hence the invariant measure (the uniform distribution on $V$) is unique. That this measure is attracting follows from the aperiodicity of $\overline{A}$.
\end{proof}

For each of the Markov chains $\left\{p_{n}(t)\right\}_{t\in\mathbb{T}_{+}}$, we define a sequence of switched Riccati iterates $\left\{P_{n}(t)\right\}$:
\begin{equation}
\label{swP}
P_{n}(t+1)=f_{p_{n}(t)}(P_{n}(t))
\end{equation}
The sequence $\left\{P_{n}(t)\right\}_{t\in\mathbb{T}_{+}}$ can be viewed as an iterated system of Riccati maps, the random switching sequence being governed by the Markov chain $\left\{p_{n}(t)\right\}_{t\in\mathbb{T}_{+}}$. A more intuitive explanation comes from the particle interpretation, precisely the $n$-th sequence may be viewed as a particle originating at node $n$ and hopping around the network as a Markov chain with transition probability $\overline{A}$ whose instantaneous state $P_{n}(t)$ evolves by the application of the Riccati operator corresponding to its current location. In particular, in contrary to the sequence of conditional error covariances at sensor $n$, $\left\{\widehat{P}_{n}(t)\right\}_{t\in\mathbb{T}_{+}}$, the sequence $\left\{P_{n}(t)\right\}_{t\in\mathbb{T}_{+}}$ does not correspond to the error evolution at a particular sensor. The following proposition shows that the sequence $\left\{P_{n}(t)\right\}_{t\in\mathbb{T}_{+}}$ is semi-Markov and establishes its relation to the sequence
$\left\{\widehat{P}_{n}(t)\right\}_{t\in\mathbb{T}_{+}}$ of interest.
\begin{proposition}
\label{prop_semi}
\item \mbox{[i]} The sequence $\left\{P_{n}(t)\right\}_{t\in\mathbb{T}_{+}}$ is semi-Markov, given the Markov switching sequence, i.e.,
{\small
\begin{equation}
\label{swP1} \mathbb{E}\left[\mathbb{I}_{\Gamma}\left(P_{n}(t+1)\right)\left|\left\{P_{n}(s),p_{n}(s)\right\}_{0\leq s\leq t}\right.\right]=\mathbb{E}\left[\mathbb{I}_{\Gamma}\left(P_{n}(t+1)\right)
\left|P_{n}(t),p_{n}(t)\right.\right],~~~\forall t\in\mathbb{T}_{+},~~\Gamma\in\mathcal{B}(\mathbb{S}_{+}^{N})
\end{equation}
}
\item \mbox{[ii]} Consider the sequence of random permutations $\left\{\pi_{t}\right\}_{t\in\mathbb{T}_{+}}$ on $V$, given by
\begin{equation}
\label{swP100}
\left(\pi_{t+1}(1),\cdots,\pi_{t+1}(N)\right)=
\left(\rightarrow\left(\pi_{t}(1),t\right),\cdots,\rightarrow\left(\pi_{t}(N),t\right)\right)
\end{equation}
with initial condition
\begin{equation}
\label{swP101}
\left(\pi_{0}(1),\cdots,\pi_{0}(N)\right)=\left(1,\cdots,N\right)
\end{equation}
(Note that $\pi_{t}(n)=p_{n}(t)$ for every $n$, where $p_{n}(t)$ is defined in eqn.~(\ref{part_n}).) Then, for $t\in\mathbb{T}_{+}$,
\begin{equation}
\label{swP2}
\left(P_{1}(t+1),\cdots,P_{N}(t+1)\right)=\left(\widehat{P}_{\pi_{t}(1)}(t+1),\cdots,\widehat{P}_{\pi_{t}(N)}(t+1)\right)
\end{equation}
\end{proposition}
Part [ii] of the above proposition suggests that the asymptotics of the desired sequence $\left\{\widehat{P}_{n}(t)\right\}_{t\in\mathbb{T}_{+}}$ for every $n$ can be
obtained by studying the same for the sequences
$\left\{P_{n}(t)\right\}_{t\in\mathbb{T}_{+}}$. Also, part~[i] of
Proposition~\ref{prop_semi} demonstrates the nice structure of the
sequence $\left\{P_{n}(t)\right\}_{t\in\mathbb{T}_{+}}$. In the following, in
particular, we will show that the sequences
$\left\{P_{n}(t)\right\}_{t\in\mathbb{T}_{+}}$ reach consensus in a weak
sense, which by part [ii] will establish weak consensus for the
sequences $\left\{\widehat{P}_{n}(t)\right\}_{t\in\mathbb{T}_{+}}$ of
interest. Hence, in the subsequent sections, we will study the
sequences $\left\{P_{n}(t)\right\}_{t\in\mathbb{T}_{+}}$, rather than working
directly on the sequences $\left\{\widehat{P}_{n}(t)\right\}_{t\in\mathbb{T}_{+}}$ of interest, which
involve a much more complicated statistical dependence.

\section{Main Results}
\label{main_res} In this section, we present and discuss the main results of the paper under Assumptions~\textbf{C.1, S.1, D.1}, see page~\pageref{assumptions.1}. The first result does not directly concern the sequences $\left\{\widehat{P}_{n}(t)\right\}$ for $n=1,\cdots,N$, but sets the stage for presenting the key result regarding the convergence of these sequences and is of independent interest.

\begin{theorem}
\label{main1}
For a given $\overline{A}$, let $\left\{\widetilde{p}(t)\right\}_{t\in\mathbb{T}_{+}}$ be a stationary Markov chain on $V$ with transition probability matrix $\overline{A}$, i.e., $\widetilde{p}(0)$ is distributed uniformly on $V$. Let $\mathbb{\nu}$ be a probability measure on $\mathbb{S}_{+}^{N}$ and consider the random process $\left\{\widetilde{P}(t)\right\}_{t\in\mathbb{T}_{+}}$ given by
    \begin{equation}
    \label{main1:2}\widetilde{P}(t+1)=f_{\widetilde{p}(t)}\left(\widetilde{P}(t)\right),~~t\in\mathbb{T}_{+}
    \end{equation}
    where $\widetilde{P}(0)$ is distributed as $\mathbb{\nu}$ and independent of the Markov chain $\left\{\widetilde{p}(t)\right\}$. Then, there exists a probability measure (unique) $\mathbb{\mu}^{\overline{A}}$ (depending on $\overline{A}$ only,) such that, for every $\mathbb{\nu}$, the process $\left\{\widetilde{P}(t)\right\}$ constructed above converges weakly to $\mathbb{\mu}^{\overline{A}}$. In other words, for any $\mathbb{\nu}$ if $\widetilde{P}(0)\sim\mathbb{\nu}$ and independent of $\left\{\widetilde{p}(t)\right\}$, we have as $t\rightarrow\infty$ that the composition of Riccati operators converges in distribution
    \begin{equation}
    \label{main1:1} f_{\widetilde{p}(t)}\circ f_{\widetilde{p}(t-1)}\cdots\circ f_{\widetilde{p}(0)}\left(\widetilde{P}(0)\right)\Longrightarrow\mu^{\overline{A}}
\end{equation}
\end{theorem}
\begin{remark} We stress here that the dependence of the invariant
measure $\mathbb{\mu}^{\overline{A}}$ on the communication policy
$\mathcal{D}$ manifests only through the mean matrix $\overline{A}$.
\end{remark}
We now state the key result characterizing the convergence properties of the sequences $\left\{\widehat{P}_{n}(t)\right\}$.
\begin{theorem}
\label{main2}
\begin{itemize}[\setlabelwidth{[i]}]

\item{\mbox{[i]}} Let $q$ be a uniformly distributed random
variable on $V$ and independent of the sequence of adjacency
matrices $\{A(t)\}_{t\in\mathbb{T}_{+}}$. Then, the sequence
$\left\{\widehat{P}_{q}(t)\right\}_{t\in\mathbb{T}_{+}}$ converges weakly to
$\mathbb{\mu}^{\overline{A}}$ (the latter being defined in
Theorem~\ref{main1}), i.e.,
    \begin{equation}
    \label{main2:1}
    \widehat{P}_{q}(t)\Longrightarrow\mu^{\overline{A}}
\end{equation}
In other words, the conditional error covariance $\left\{\widehat{P}_{q}(t)\right\}$ of any randomly selected sensor (estimator) converges in distribution to $\mathbb{\mu}^{\overline{A}}$.
\item{\mbox{[ii]}} For every $n\in [1,\cdots,N]$, the sequence
$\left\{P_{n}(t)\right\}_{t\in\mathbb{T}_{+}}$ (or the sequence
$\left\{\widehat{P}_{\pi_{t}(n)}(t)\right\}_{t\in\mathbb{T}_{+}}$ is
stochastically dominated by the distribution
$\mathbb{\mu}^{\overline{A}}$ as $t\rightarrow\infty$, i.e., for
every $\alpha>0$, we have
{\small
    \begin{equation}
    \label{main2:2}\limsup_{t\rightarrow\infty}
    \mathbb{P}\left(\left\|P_{n}(t)\right\|\geq\alpha\right)
    \leq\mathbb{\mu}^{\overline{A}}\left(\left\{X\in\mathbb{S}_{+}^{N}\left|\left\|X\right\|
    \geq\alpha\right.\right\}\right)
    \end{equation}
    \begin{equation}
    \label{main2:3}\limsup_{t\rightarrow\infty}\mathbb{P}\left(P_{n}(t)\succeq\alpha I\right)\leq\mathbb{\mu}^{\overline{A}}\left(\left\{X\in\mathbb{S}_{+}^{N}\left|X\succeq\alpha I\right.\right\}\right)
    \end{equation}
    }
    More generally, for a closed set $F$ preserving monotonicity, i.e., $X\in F$ implies $Y\in F$ for all $Y\succeq X$, we have
    \begin{equation}
    \label{main2:4}
    \limsup_{t\rightarrow\infty}\mathbb{P}\left(P_{n}(t)\in F\right)\leq\mathbb{\mu}^{\overline{A}}\left(F\right)
    \end{equation}
    In words, $\forall n$, the pathwise error associated with $\widehat{\mathbf{x}}_{\pi_{t}(n)}(t)$ is stochastically dominated by
    $\mathbb{\mu}^{\overline{A}}$.
\item{\mbox{[iii]}} For each $n$, the sequence of error
covariances $\left\{\widehat{P}_{n}(t)\right\}_{t\in\mathbb{T}_{+}}$
is stochastically bounded,
\begin{equation}
\label{main2:18}
\lim_{J\rightarrow\infty}\sup_{t\in\mathbb{T}_{+}}\mathbb{P}\left(\left\|\widehat{P}_{n}(t)\right\|\geq
J\right)=0
\end{equation}
Specifically, for all closed $F$, we have
\begin{equation}
\label{main2:19}
\limsup_{t\rightarrow\infty}\mathbb{P}\left(\widehat{P}_{n}(t)\in
F\right)\leq N\mathbb{\mu}^{\overline{A}}\left(F\right)
\end{equation}
\end{itemize}
\end{theorem}
We discuss the consequences of Theorem~\ref{main2}. The first part
of the theorem reinforces the weak consensus achieved by the GIKF
algorithm, i.e., the conditional error covariance at a randomly
selected sensor converges in distribution to the invariant measure
$\mathbb{\mu}^{\overline{A}}$. Reinterpreted, it provides an
estimate $\{\widehat{\mathbf{x}}_{q}(t)\}$ (in practice, obtained
by uniformly selecting a sensor $q$ independent of the random
gossip protocol $\{A(t)\}$ and using its estimate
$\widehat{\mathbf{x}}_{q}(t)$ for all time $t$) with
stochastically bounded conditional error covariance under the weak
detectability and connectivity assumptions. Note that the results
provided in this paper pertain to the limiting distribution of the
conditional error covariance and, hence, the pathwise filtering
error. This is a much stronger result than providing moment
estimates of the conditional error covariance, which does not
provide much insight into the pathwise instantiations of the
filter. In this paper, we do not provide analytic
characterizations of the resulting invariant measure
$\mathbb{\mu}^{\overline{A}}$. However, Theorem~\ref{main1} also
provides an efficient numerical characterization of
$\mathbb{\mu}^{\overline{A}}$. In particular, the weak convergence
in eqn.~(\ref{main1:1}) shows that the empirical distribution
obtained by plotting repeated instantiations of the process
$\left\{\widetilde{P}(t)\right\}$ (eqn.~(\ref{main1:2})) would converge to
$\mathbb{\mu}^{\overline{A}}$.

Another class of estimates obtained by the GIKF algorithm is
demonstrated in the second part of Theorem~\ref{main2}. For each
$n$, the estimate $\left\{\widehat{\mathbf{x}}_{\pi_{t}(n)}(t)\right\}$ is
obtained in practice by starting at the node $n$ and then
performing a random walk, $\pi_{t}(n)$, through the graph and
collecting the estimates on the way. Eqns.~(\ref{main2:2}-\ref{main2:4})
show that, in the limit as $t\rightarrow\infty$, these estimates
are at least as good as the estimate
$\left\{\widehat{\mathbf{x}}_{q}(t)\right\}$ obtained by probing a randomly
selected node and using its estimate throughout. For some $n$,
whether the estimate $\left\{\widehat{\mathbf{x}}_{\pi_{t}(n)}(t)\right\}$ is
strictly better than the estimate
$\left\{\widehat{\mathbf{x}}_{q}(t)\right\}$ asymptotically is an interesting
technical question and not resolved in this paper. On the
contrary, another possibility may be an extension of
eqn.~(\ref{main2:4}) to all closed $F$ leading to the weak
convergence of $\left\{\widehat{P}_{\pi_{t}(n)}(t)\right\}$ to
$\mathbb{\mu}^{\overline{A}}$ by Portmanteau's theorem. However,
the inequality in eqn.~(\ref{main2:4}) cannot be strict for all
$n$, as we have for all closed $F$ and $\varepsilon>0$ (see
Subsection~\ref{pr_main2},) \begin{equation}
\label{main2:80}
\frac{1}{N}\sum_{n=1}^{N}\liminf_{t\rightarrow\infty}\mathbb{P}\left(\widehat{P}_{\pi_{t}(n)}(t)\in
F^{\varepsilon}\right)\geq\mathbb{\mu}^{\overline{A}}
\end{equation}

The last part of Theorem~\ref{main2} shows that weak detectability
(which is necessary for the error of a centralized estimator to be stochastic bounded) is sufficient in the distributed gossip setting to lead to sensor estimates with stochastically bounded
errors. The upper bound presented in eqn.~(\ref{main2:19}) is
highly conservative and in fact, we have for all closed $F$ (see
Subsection~\ref{pr_main2},) \begin{equation}
\label{main2:81}
\sum_{n=1}^{N}\limsup_{t\rightarrow\infty}\mathbb{P}\left(\widehat{P}_{n}(t)\in
F\right)\leq N\mathbb{\mu}^{\overline{A}}
\end{equation}
%
%
%
%
%
%
\section{The auxiliary sequence $\{\widetilde{P}_{t}\}$: RDS formulation}
\label{aux}
The asymptotic analysis of the semi-Markov processes $\left\{P_{n}(t)\right\}$
for $n=1,\cdots, N$ does not fall under the purview of standard
approaches based on iterated random systems
(\cite{DiaconisFreedman}) or a random dynamical
system (RDS) (\cite{ArnoldChueshovbook}) as the switching Markov chains
$\left\{p_{n}(t)\right\}$ are non-stationary. In this section, we consider an
auxiliary process $\left\{\widetilde{P}(t)\right\}$ whose evolution is
governed by similar random Riccati iterates, the difference being that
the switching Markov chain is stationary i.e., the switching
Markov chain $\left\{\widetilde{p}(t)\right\}$ is initialized with the
uniform invariant measure on $V$. We analyze the asymptotic
properties of the auxiliary sequence $\left\{\widetilde{P}(t)\right\}$ by
formulating it as a RDS on the space $\mathbb{S}_{+}^{N}$ and then
in subsequent sections we derive the asymptotics of the sequences
$\left\{P_{n}(t)\right\}$ for $n=1,\cdots, N$ through comparison arguments.
We start by formally defining the sequence
$\left\{\widetilde{P}(t)\right\}$\footnote{We are
interested in the distributional properties of the various
processes of concern. The actual pathwise construction is not of
importance as long as the required distributional equivalence
holds. We assume that the measure space
$\left(\Omega,\mathcal{F},\mathbb{P}\right)$ is rich enough (or
suitably extended) to carry out constructions of the various
auxiliary random variables.}:

Consider a Markov chain on the graph $V$, $\left\{\widetilde{p}(t)\right\}_{t\in\mathbb{T}_{+}}$, with transition matrix $\overline{A}$ and uniform initial distribution, i.e.,
\begin{equation}
\label{aux1}
\mathbb{P}\left[\widetilde{p}(0)=n\right]=\frac{1}{N},~~~n=1,\cdots,N
\end{equation}
By Proposition~\ref{prop_int}, the Markov chain $\left\{\widetilde{p}(t)\right\}$ is stationary.

We now define the auxiliary process $\left\{\widetilde{P}(t)\right\}$ as follows:
\begin{equation}
\label{aux2}
\widetilde{P}(t+1)=f_{\widetilde{P}(t)}\left(\widetilde{P}(t)\right)
\end{equation}
with (possibly random) initial condition $\widetilde{P}(0)$.\footnote{Although the sequences $\left\{P_{n}(t)\right\}$ of interest have deterministic initial conditions, it is required for technical reasons (to be made precise later) to allow random initial states $\widetilde{P}(0)$, when studying the auxiliary sequence $\left\{\widetilde{P}(t)\right\}$.}

 Before reading the next two Sections, we refer the reader to Appendix~\ref{RDS} where we review preliminary facts and results from the theory of monotone, sublinear random dynamical systems (RDS) (\cite{Chueshov}) tailored to our needs. We then show in Subsection~\ref{RDS_aux} that the sequence $\left\{\widetilde{P}(t)\right\}$, for each $n$, admits an ergodic RDS formulation evolving on $\mathbb{S}_{+}^{N}$ and establish some of its properties in Subsection~\ref{RDS_aux-2}.
\subsection{RDS formulation of $\left\{\widetilde{P}(t)\right\}$}
\label{RDS_aux} In this subsection, we construct a RDS $(\theta^{R},\varphi^{R})$ on $\mathbb{S}_{+}^{N}$, which is equivalent to the auxiliary sequence $\left\{\widetilde{P}(t)\right\}$ in distribution. To this end, we construct the Markov chain $\left\{\widetilde{p}(t)\right\}$ (in a distributional sense) on the canonical path space. Let $\widetilde{\Omega}$ denote the set $\{1,\cdots,N\}$ with $\widetilde{\mathcal{F}}$ denoting the corresponding Borel algebra on $\widetilde{\Omega}$, which coincides with the power set of $\{1,\cdots,N\}$. Denote by $\Omega^{R}$ the two-sided infinite product of sets $\widetilde{\Omega}$, $\Omega^{R}=\bigotimes_{t=-\infty}^{\infty}\widetilde{\Omega}$, i.e., $\Omega^{R}$ is the space of double-sided sequences of entries in $\{1,\cdots,N\}$, i.e.,
\begin{equation}
\label{RDS_aux1}
\Omega^{R}=\left\{\omega=\left(\cdots,\omega_{-1},\omega_{0},\omega_{1},\cdots\right)\left|\omega_{t}\in\{1,\cdots,N\},~~\forall t\in\mathbb{T}\right.\right\}
\end{equation}
We equip $\Omega^{R}$ with the corresponding product Borel algebra $\mathcal{F}^{R}=\bigotimes_{t=-\infty}^{\infty}\widetilde{\mathcal{F}}$ generated by the cylinder sets. Note that $\left\{\omega_{t}\right\}_{t\in\mathbb{T}_{+}}$ for all $\omega\in\Omega^{R}$ denotes the canonical path space (trajectory) of the Markov chain $\left\{\widetilde{p}(t)\right\}_{t\in\mathbb{T}_{+}}$. The reason for introducing two-sided sequences is a matter of technical convenience and will be evident soon. Consider the unique probability measure $\mathbb{P}^{R}$ on $\mathcal{F}^{R}$, under which the stochastic process (two-sided) $\left\{\omega_{t}\right\}_{t\in\mathbb{T}}$ is a stationary Markov chain on the finite state space $\{1,\cdots,N\}$ with transition probability matrix $\overline{A}$. By the assumption of stationarity and Proposition~\ref{prop_int}, the distribution of $\omega_{t}$ for each $t\in\mathbb{T}$ is necessarily the uniform distribution on $\{1,\cdots,N\}$. In particular, we note that the stochastic processes $\left\{\widetilde{p}(t)\right\}_{t\in\mathbb{T}}$ and $\{\omega_{t}\}_{t\in\mathbb{T}}$ are equivalent in terms of the distribution induced on path space. Define the family of transformations
$\left\{\theta^{R}_{t}\right\}_{t\in\mathbb{T}}$ on $\Omega$ as the family of
left-shifts, i.e,
\begin{equation}
\label{def_RDS3} \theta_{t}^{R}\omega=\omega(t+\cdot),~~\forall
t\in\mathbb{T}
\end{equation}
With this, the space
$\left(\Omega^{R},\mathcal{F}^{R},\mathbb{P}^{R},\left\{\theta_{t}^{R},t\in\mathbb{T}\right\}\right)$
becomes the canonical path space of a two-sided stationary sequence equipped with the left-shift operator and
hence (see, for example, \cite{Kallenberg}) satisfies the Assumptions~\textbf{A.1)-A.3)} in Definition~\ref{defn_RDS} to be a metric dynamical system and, in fact, is also ergodic.

We now set to define the cocycle $\varphi^{R}$, see also Definition~\ref{defn_RDS}, over $\mathbb{S}_{+}^{N}$, which gives the RDS of interest. We define $\varphi^{R}:\mathbb{T}_{+}\times\Omega^{R}\times\mathbb{S}_{+}^{N}\longmapsto\mathbb{S}_{+}^{N}$ by:
{\small
\begin{eqnarray}
\label{def_RDS5} \varphi^{R}(0,\omega,X)&=&X,~~\forall\omega,X
\\
\label{RDS6}
\varphi^{R}(1,\omega,X)&=&f_{\omega_{0}}(X),~~\forall\omega,X
\\
\label{def_RDS7}
\varphi^{R}(t,\omega,X)&=&f_{\theta^{R}_{t-1}\omega(0)}\left(\varphi^{R}(t-1,\omega,X)\right)
=f_{\omega_{t-1}}\left(\varphi^{R}(t-1,\omega,X)\right),~~\forall
t>1,\omega,X
\end{eqnarray}
}
(Note that, by property of the left shift $\theta^{R}$, we have $\theta^{R}_{t-1}\omega(0)=\omega_{t}$, which explains the equality in eqn.~(\ref{def_RDS7}).) The cocycle $\varphi^{R}$ defined satisfies the assumptions of measurability jointly in its arguments, and the continuity of the map $\varphi^{R}(t,\omega,\cdot):\mathbb{S}_{+}^{N}\longmapsto\mathbb{S}_{+}^{N}$ w.r.t.~the phase variable $X$ for each fixed $t,\omega$ follows from the continuity of the corresponding Riccati operator. The pair $\left(\theta^{R},\varphi^{R}\right)$ thus forms a well-defined RDS on the phase space $\mathbb{S}_{+}^{N}$. Now consider the sequence of random variables $\left\{\varphi^{R}(t,\omega,P_{n}(0))\right\}_{t\in\mathbb{T}_{+}}$ (as explained earlier, the randomness is induced by $\omega$,) which can be viewed as successive (random) iterates of the RDS $\left(\theta^{R},\varphi^{R}\right)$ starting with the initial state $P_{n}(0)$. By construction, it follows that the sequence $\left\{\varphi^{R}(t,\omega,P_{n}(0))\right\}_{t\in\mathbb{T}_{+}}$ is distributionally equivalent to the sequence $\left\{\widetilde{P}(t)\right\}_{t\in\mathbb{T}_{+}}$. In particular,
\begin{equation}
\label{def_RDS9}
\varphi^{R}\left(t,\omega,P_{n}(0)\right)\ndtstile{}{d}\widetilde{P}(t),~~~\forall t\in\mathbb{T}_{+}
\end{equation}
Thus, analyzing the asymptotic distributional properties of the sequence $\left\{\widetilde{P}(t)\right\}_{t\in\mathbb{T}_{+}}$ is equivalent to studying the sequence $\left\{\varphi^{R}\left(t,\omega,P_{n}(0)\right)\right\}_{t\in\mathbb{T}_{+}}$, which we undertake in the next subsection.

\subsection{Properties of the RDS $(\theta^{R},\varphi^{R})$}
\label{RDS_aux-2} We establish some basic properties of the RDS $\left(\theta^{R},\varphi^{R}\right)$ representing the auxiliary sequence $\left\{\widetilde{P}(t)\right\}$.
\begin{lemma}
\label{prop_auxRDS}
\item\mbox{[i]} The RDS $\left(\theta^{R},\varphi^{R}\right)$ is conditionally compact.
\item\mbox{[ii]} The RDS $\left(\theta^{R},\varphi^{R}\right)$ is order preserving.
\item\mbox{[iii]} If in addition $\mathcal{Q}$ is positive definite, i.e., $\mathcal{Q}\gg 0$, then $\left(\theta^{R},\varphi^{R}\right)$ is strongly sublinear.
\end{lemma}
\begin{proof} The claim in [i] (conditional compactness) is an immediate consequence of the finite dimensionality of the underlying vector space $\mathbb{S}_{+}^{N}$.

The order preserving property [ii] follows from the monotonicity of the individual Riccati operators $f_{n}$ and hence finite compositions of them remain order-preserving.

The strong sublinearity uses the concavity of the Riccati operators and their monotone nature and is a routine extension to an arbitrary number $N$ of Riccati operators, given the development in~\cite{Riccati-weakconv} (see Lemma 21 in~\cite{Riccati-weakconv}) for the case of two Riccati operators.
\end{proof}
\section{Asymptotics of $\left\{\widetilde{P}(t)\right\}$}
\label{RDS_auxproof} The main result here concerns the asymptotic properties of the auxiliary sequences $\left\{\widetilde{P}(t)\right\}_{t\in\mathbb{T}_{+}}$ for each $n\in [1,\cdots,N]$. We have the following:
\begin{theorem}
\label{conv_aux} Under the assumptions~\textbf{C.1,S.1,D.1}, see page~\pageref{assumptions.1}, there exists a unique equilibrium probability measure $\mathbb{\mu}^{\overline{A}}$ on the space of positive semidefinite matrices $\mathbb{S}_{+}^{N}$, such that, for each $n\in [1,\cdots,N]$, the sequence $\left\{\widetilde{P}(t)\right\}_{t\in\mathbb{T}_{+}}$ converges weakly (in distribution) to $\mathbb{\mu}^{\overline{A}}$ from every initial condition $P_{n}(0)$:
\begin{equation}
\label{conv_aux1}\left\{\widetilde{P}(t)\right\}\Longrightarrow\mathbb{\mu}^{\overline{A}},~~~\forall n\in [1,\cdots,N]
\end{equation}
\end{theorem}
The rest of the subsection is devoted to the proof of the above result. But, before that, we highlight some consequences of Theorem~\ref{conv_aux}.
\begin{remark} It is important to note, as stated in Theorem~\ref{conv_aux}, that the equilibrium measure $\mathbb{\mu}^{\overline{A}}$ does not depend on the index $n$ and the initial state $\widetilde{P}(0)$ of the sequence $\left\{\widetilde{P}(t)\right\}$, but is a functional of the network topology and the particular (randomized) communication protocol captured by the matrix $\overline{A}$. Theorem~\ref{conv_aux}, thus concludes that the sequences $\left\{\widetilde{P}(t)\right\}$ reach consensus \emph{in the weak sense} to the same equilibrium measure irrespective of the initial states.
\end{remark}
The proof of Theorem~\ref{conv_aux} is rather long and technical, which we accomplish in steps.

\begin{lemma}
\label{dist_obsRiccati}
Recall Assumption~\textbf{D.1}, page~\pageref{assumptiond.2}, and let, in particular, $w_{0}=(n_{1},\cdots,n_{\ell})$ be a walk such that, the Grammian
\begin{equation}
\label{unif1}G_{w_{0}}=\sum_{i=1}^{\ell}\left(\mathcal{F}^{i-1}\right)^{T}\mathcal{C}_{n_{i}}^{T}\mathcal{C}_{n_{i}}\mathcal{F}^{i-1}
\end{equation}
is invertible, where $\ell\geq 1$. Define the function $g_{w_{0}}:\mathbb{S}_{+}^{N}\longmapsto\mathbb{S}_{+}^{N}$ by
\begin{equation}
\label{unif3}g_{w_{0}}(X)=f_{n_{\ell}}\circ f_{n_{\ell}-1}\circ\cdots\circ f_{n_{1}}\left(X\right)
\end{equation}
Then, there exists a constant $\alpha_{0}>0$ such that the following uniformity condition holds:
\begin{equation}
\label{unif2}
g_{w_{0}}(X)\leq \alpha_{0} I,~~~\forall X\in\mathbb{S}_{+}^{N}
\end{equation}
In other words the iterate $g_{w_{0}}(\cdot)$ is uniformly bounded irrespective of the value of the argument.
\end{lemma}
The proof is provided in Appendix~\ref{RDS_123}. Note that in eqn.~(\ref{unif1}) the observation matrix~$\mathcal{C}_{n_{i}}$ is indexed by~$n_{i}$ the current site visited by the random walk~$w_0$ introduced in Lemma~\ref{dist_obsRiccati}. Also, note that the function~$g_{w_{0}}(X)$ defined in eqn.~(\ref{unif3}) is indexed by the walk~$w_0$.

The following key lemma establishes asymptotic boundedness properties of $\left\{\widetilde{P}(t)\right\}$ and is proved in Appendix~\ref{RDS_123}.
\begin{lemma}
\label{stoch_boundedness} The sequence $\left\{\widetilde{P}(t)\right\}$ is stochastically bounded for each $n$ under the Assumptions of Theorem~\ref{conv_aux}, i.e.,
\begin{equation}
\label{sb1}
\lim_{J\rightarrow\infty}\sup_{t\in\mathbb{T}_{+}}\mathbb{P}\left(\left\|\widetilde{P}(t)\right\|>J\right)=0
\end{equation}
\end{lemma}

We now complete the proof of Theorem~\ref{conv_aux}.
From Lemma~\ref{RDS_aux} we note that
$\left(\theta^{R},\varphi^{R}\right)$ is strongly sublinear, conditionally
compact and order-preserving. Also, the cone $\mathbb{S}_{+}^{N}$
satisfies the conditions required in the hypothesis of
Theorem~\ref{LSD}. We note for
$t>0$
\begin{equation}
\label{step2} \varphi^{R}\left(t,\omega,0\right) =
f_{\omega(t-1)}\left(\varphi(t-1,\omega,0)\right)
\succeq \mathcal{Q} \gg  0
\end{equation}
Thus the hypotheses of Theorem~\ref{LSD} are satisfied, and
precisely one of the assertions~\textbf{a)} and~\textbf{b)} holds.
By an argument similar to Lemma~23 in~\cite{Riccati-weakconv}, we can show that assertion~\textbf{a)} cannot hold in the face of stochastic boundedness of the sequence $\left\{\widetilde{P}(t)\right\}_{t\in\mathbb{T}_{+}}$ (Lemma~\ref{stoch_boundedness}). Thus assertion~\textbf{b)}
holds, and, as a direct consequence of Theorem~\ref{LSD}, we establish the existence of a unique almost equilibrium
$u^{\overline{A}}(\omega)\gg 0$ defined on a
$\theta^{R}$-invariant set $\Omega^{\ast}\in\mathcal{F}^{R}$ with
$\mathbb{P}\left(\Omega^{\ast}\right)=1$ such
that, for any random variable $v(\omega)$ possessing the property
$0\preceq v(\omega)\preceq\alpha u^{\overline{A}}(\omega)$
for all $\omega\in\Omega^{\ast}$ and deterministic $\alpha>0$, the
following holds:
\begin{equation}
\label{step1}
\lim_{t\rightarrow\infty}\varphi\left(t,\theta_{-t}\omega,v(\theta_{-t}\omega)\right)=
u^{\overline{A}}(\omega),~~\omega\in\Omega^{\ast}
\end{equation}
From the distributional equivalence of pull-back and forward orbits, Lemma~\ref{orbit_lemma} establishes the existence of a unique almost equilibrium $u^{\overline{A}}$, i.e., a unique equilibrium measure for the process $\left\{\widetilde{P}(t)\right\}$ from the distributional equivalence of pull-back and forward orbits. However, to show that the measure induced by $u^{\overline{A}}$ on $\mathbb{S}^{N}_{+}$ is attracting for $\left\{\widetilde{P}(t)\right\}$,
eqn.~(\ref{step1}) must hold for all initial $v$, whereas
Lemma~\ref{orbit_lemma} establishes convergence for a restricted class of
initial conditions $v$. We need the following result to extend it
to general initial conditions.

\begin{lemma}
\label{eq} Under the assumptions of Theorem~\ref{conv_aux}, let
$u^{\overline{A}}$ be the unique almost equilibrium of the RDS
$\left(\theta^{R},\varphi^{R}\right)$. Then
\begin{equation}
\label{eq1}
\mathbb{P}\left(\omega:u^{\overline{A}}(\omega)\succeq
\mathcal{Q}\right)=1
\end{equation}
\end{lemma}
\begin{proof}
The proof uses the fact that, for all $n$, $f_{n}(X)\succeq \mathcal{Q}$,
and is routine given the corresponding development in Lemma~24 of~\cite{Riccati-weakconv}.
\end{proof}

We now complete the proof of Theorem~\ref{conv_aux}.

\begin{proof}[Proof of Theorem~\ref{conv_aux}]: Let $\mu^{\overline{A}}$
be the distribution of the unique almost equilibrium in
eqn.~(\ref{step1}). By Lemma~\ref{eq} we have $\mu^{\overline{A}}\left(\mathbb{S}_{++}^{N}\right)=1$.
Let $P_{0}\in\mathbb{S}_{+}^{N}$ be an arbitrary initial state. By
construction of the RDS $\left(\theta^{R},\varphi^{R}\right)$, the sequences
$\left\{P_{t}\right\}_{t\in\mathbb{T}_{+}}$ and
$\left\{\varphi^{R}\left(t,\omega,P_{0}\right)\right\}_{t\in\mathbb{T}_{+}}$
are distributionally equivalent, i.e., $P_{t}\ndtstile{}{d}\varphi^{R}\left(t,\omega,P_{0}\right)$.
Recall $\Omega^{\ast}$ as the $\theta^{R}$-invariant set with
$\mathbb{P}^{\overline{\gamma}}\left(\Omega^{\ast}\right)=1$ in
eqn.~(\ref{step1}) on which the almost equilibrium
$u^{\overline{A}}$ is defined. By Lemma~\ref{eq}, there
exists $\Omega_{1}\subset\Omega^{\ast}$ with
$\mathbb{P}^{\overline{\gamma}}\left(\Omega_{1}\right)=1$, such that
\begin{equation}
\label{main_res5} u^{\overline{A}}(\omega)\succeq
\mathcal{Q},~~\omega\in\Omega_{1}
\end{equation}
Define the random variable
$\widetilde{X}:\Omega\longmapsto\mathbb{S}_{+}^{N}$ by
\begin{equation}
\label{main_res6} \left\{ \begin{array}{ll}
                    P_{0} & \mbox{if $\omega\in\Omega_{1}$} \\
                    0 & \mbox{if $\omega\in\Omega_{1}^{c}$}
                   \end{array}
          \right.
\end{equation}
Now choose $\alpha>0$ sufficiently large such that
\begin{equation}
\label{main_res7} P_{0}\preceq\alpha \mathcal{Q}
\end{equation}
This is possible because $\mathcal{Q}\gg 0$. Then
\begin{equation}
\label{main_res8} 0\preceq\widetilde{X}(\omega)\preceq\alpha
u^{\overline{A}}(\omega),~~\omega\in\Omega^{\ast}
\end{equation}
Indeed, we have
\begin{equation}
\label{main_res9} 0\preceq
P_{0}=\widetilde{X}(\omega)\preceq\alpha \mathcal{Q}\preceq\alpha
u^{\overline{A}}(\omega),~~\omega\in\Omega_{1}
\end{equation}
and
\begin{equation}
\label{main_res10} 0=\widetilde{X}(\omega)\preceq\alpha
u^{\overline{A}}(\omega),~~\omega\in\Omega\backslash\Omega_{1}
\end{equation}
We then have by the discussion preceding eqn.~(\ref{step1})
\begin{equation}
\label{main_res11}
\lim_{t\rightarrow\infty}\varphi^{R}\left(t,\theta_{-t}\omega,\widetilde{X}\left(\theta_{-t}\omega\right)\right)
=u^{\overline{A}}(\omega),~~\omega\in\Omega^{\ast}
\end{equation}
Since convergence $\mathbb{P}^{\overline{\gamma}}$~a.s. implies convergence in distribution, we have
\begin{equation}
\label{main_res12}
\varphi^{R}\left(t,\theta_{-t}\omega,\widetilde{X}\left(\theta_{-t}\omega\right)\right)\Longrightarrow\mu^{\overline{A}}
\end{equation}
as $t\rightarrow\infty$, where $\Longrightarrow$ denotes weak convergence or convergence in distribution. Then, by
Lemma~\ref{orbit_lemma}, the sequence $\left\{\varphi^{R}\left(t,\omega,\widetilde{X}(\omega)\right)\right\}_{t\in\mathbb{T}_{+}}$
also converges in distribution to the unique stationary distribution $\mu^{\overline{A}}$, i.e., as
$t\rightarrow\infty$
\begin{equation}
\label{main_res13}
\varphi^{R}\left(t,\omega,\widetilde{X}(\omega)\right)\Longrightarrow\mu^{\overline{A}}
\end{equation}
Now, since $\mathbb{P}^{\overline{\gamma}}\left(\Omega_{1}\right)=1$, by
eqn.~(\ref{main_res6})
\begin{equation}
\label{main_res14}
\varphi^{R}\left(t,\omega,P_{0}\right)=\varphi^{R}\left(t,\omega,\widetilde{X}(\omega)\right),~~\mathbb{P}^{\overline{\gamma}}~a.s.,~t\in\mathbb{T}_{+}
\end{equation}
which implies
\begin{equation}
\label{main_res15}
\varphi^{R}\left(t,\omega,P_{0}\right)\ndtstile{}{d}\varphi^{R}\left(t,\omega,\widetilde{X}(\omega)\right),~~t\in\mathbb{T}_{+}
\end{equation}
From eqns.~(\ref{main_res13},\ref{main_res15}), we then have $\varphi^{R}\left(t,\omega,P_{0}\right)\Longrightarrow\mu^{\overline{A}}$,
which together with the distributional equivalence
$P_{t}\ndtstile{}{d}\varphi^{R}\left(t,\omega,P_{0}\right)$ noted above
 implies, as $t\rightarrow\infty$, $P_{t}\Longrightarrow\mu^{\overline{A}}$.
\end{proof}

\section{Proofs of main results}
\label{proof_main_res}
\subsection{Proof of Theorem~\ref{main1}}
\label{pr_main1}
\begin{proof}
By Theorem~\ref{conv_aux} we know that such a sequence $\left\{\widetilde{P}(t)\right\}$ converges weakly to $\mathbb{\mu}^{\overline{A}}$ when started from a deterministic initial condition. In the case, $\widetilde{P}(0)$ is distributed as $\mathbb{\nu}$, we note that, by the independence of $\widetilde{P}(0)$ and the Markov chain $\{q(t)\}$,
\begin{equation}
\label{pr_main1:1}
\mathbb{E}\left[g\left(\widetilde{P}(t)\right)\right]=\int_{\mathbb{S}_{+}^{N}}
\mathbb{E}\left[\left(\widetilde{P}(t)\right)\left|\widetilde{P}(0)=X\right.\right]d\mathbb{\nu}(X)
\end{equation}
for any $g\in C_{b}(\mathbb{S}_{+}^{N})$. Now, the distribution of the sequence $\left\{\widetilde{P}(t)\right\}$ conditioned on the event $\widetilde{P}(0)=X$ is the same as that when the sequence starts with the deterministic initial condition $X$ (this is true because $\widetilde{P}(0)$ is independent of $\{q(t)\}$.) Hence by Theorem~\ref{conv_aux}
\begin{equation}
\label{pr_main1:2}
\lim_{t\rightarrow\infty}\mathbb{E}\left[\left(\widetilde{P}(t)\right)\left|\widetilde{P}(0)=X\right.\right]=
\int_{\mathbb{S}_{+}^{N}}g(y)d\mathbb{\mu}^{\overline{A}}(Y)
\end{equation}
for all $X$. Since $g$ is bounded, the dominated convergence theorem and eqn.~(\ref{pr_main1:1}) lead to
\begin{equation}
\label{pr_main1:3}
\lim_{t\rightarrow\infty}\mathbb{E}\left[\left(\widetilde{P}(t)\right)\right]=
\int_{\mathbb{S}_{+}^{N}}g(y)d\mathbb{\mu}^{\overline{A}}(Y)
\end{equation}
for all $g\in C_{b}(X)$, and hence the required weak convergence follows.
\end{proof}

\subsection{Proof of Theorem~\ref{main2}}
\label{pr_main2}
\begin{proof} We prove Theorem~\ref{main2} in the order \textbf{1)},\textbf{3)} and \textbf{2)}.

Consider any $\Gamma\in\mathcal{B}(\mathbb{S}_{+}^{N})$. We
estimate the probability
$\mathbb{P}\left(\widehat{P}_{q}(t)\in\Gamma\right)$. To this end, we note that
\begin{equation}
\label{pr_main2:1}
\mathbb{P}\left(\widehat{P}_{q}(t)\in\Gamma\right) =  \sum_{n=1}^{N}\mathbb{P}\left(\widehat{P}_{n}(t)\in\Gamma\right)\mathbb{P}\left(q=n\right)=\frac{1}{N}\sum_{n=1}^{N}\mathbb{P}\left(\widehat{P}_{n}(t)\in\Gamma\right)
\end{equation}
The first step holds because $q$ is independent of the sequences $\left\{\widehat{P}_{n}(t)\right\}$ for all $n$ and subsequently we use that $q$ is uniformly distributed on $V$. Denoting by $\pi^{-1}_{t}$ the inverse of the permutation $\pi_{t}$, we have
\begin{equation}
\label{pr_main2:2}
\mathbb{P}\left(\widehat{P}_{n}(t)\in\Gamma\right) = \mathbb{P}\left(P_{\pi^{-1}_{t}(n)}(t)\in\Gamma\right) = \sum_{l=1}^{N}\mathbb{P}\left(\left\{P_{l}(t)\in\Gamma\right\}\bigcap\left\{\pi^{-1}_{t}(n)=l\right\}\right)
\end{equation}
Note, here, unlike in eqn.~(\ref{pr_main2:1}), we may not gain much by splitting the probabilities in the last step as the events $\left\{P_{l}(t)\in\Gamma\right\}$ and $\left\{\pi^{-1}_{t}(n)=l\right\}$ are not independent. Combining eqns.~(\ref{pr_main2:1},\ref{pr_main2:2}), we have
{\small
\begin{eqnarray}
\label{pr_main2:3}
\mathbb{P}\left(\widehat{P}_{q}(t)\in\Gamma\right) & = & \frac{1}{N}\sum_{n=1}^{N}\sum_{l=1}^{N}\mathbb{P}\left(\left\{P_{l}(t)\in\Gamma\right\}
\bigcap\left\{\pi^{-1}_{t}(n)=l\right\}\right)\nonumber \\ & = & \frac{1}{N}\sum_{l=1}^{N}\sum_{n=1}^{N}\mathbb{P}\left(\left\{P_{l}(t)\in\Gamma\right\}
\bigcap\left\{\pi^{-1}_{t}(n)=l\right\}\right)\nonumber \\ & = & \frac{1}{N}\sum_{l=1}^{N}\mathbb{P}\left(P_{l}(t)\in\Gamma\right)
\end{eqnarray}
}
Note the last step follows from the fact that
\begin{equation}
\label{pr_main2:4}
\sum_{n=1}^{N}\mathbb{P}\left(\left\{P_{l}(t)\in\Gamma\right\}
\bigcap\left\{\pi^{-1}_{t}(n)=l\right\}\right)=\mathbb{P}\left(P_{l}(t)\in\Gamma\right)
\end{equation}
because the events $\left\{\pi^{-1}_{t}(n)=l\right\}$, $n=1,\cdots,N$ are mutually exclusive and exhaustive, $\pi^{-1}(t)$ being a permutation.

Now consider a stationary Markov chain $\left\{\widetilde{p}(t)\right\}$ on $V$ with transition probability $\overline{A}$ and let $\left\{\widetilde{P}(t)\right\}$ be the sequence defined by
\begin{equation}
\label{pr_main2:5}
\widetilde{P}(t+1)=f_{\widetilde{p}(t)}\left(\widetilde{P}(t)\right),~~t\in\mathbb{T}_{+}
\end{equation}
with initial condition $\widetilde{P}(0)=\widehat{P}(0)$. Then,
\begin{equation}
\label{pr_main2:6}\mathbb{P}\left(\widetilde{P}(t)\in\Gamma\right)  = \sum_{l=1}^{N}\mathbb{P}\left(\left.\widetilde{P}(t)\in\Gamma\right|\widetilde{p}(0)=
l\right)\mathbb{P}\left(\widetilde{p}(0)=l\right) =  \frac{1}{N}\mathbb{P}\left(\left.\widetilde{P}(t)\in\Gamma\right|\widetilde{p}(0)=l\right)
\end{equation}
By construction, the distribution of the sequence $\left\{\widetilde{P}(t)\right\}$ conditioned on the event $\{\widetilde{p}(0)=l\}$ is equivalent to that of the sequence $\{P_{l}(t)\}$ and hence
\begin{equation}
\label{pr_main2:7}
\mathbb{P}\left(\left.\widetilde{P}(t)\in\Gamma\right|\widetilde{p}(0)=l\right)=\mathbb{P}\left(P_{l}(t)\in\Gamma\right)
\end{equation}
Hence by eqns.~(\ref{pr_main2:3},\ref{pr_main2:7}) we obtain
\begin{equation}
\label{pr_main2:8}\mathbb{P}\left(\widehat{P}_{q}(t)\in\Gamma\right) =\mathbb{P}\left(\widetilde{P}(t)\in\Gamma\right)
\end{equation}
Thus, for all $t$, $P_{q}(t)\ndtstile{}{d}\widetilde{P}(t)$.
By Theorem~\ref{conv_aux}, we then have the weak convergence of the sequence $\{P_{q}(t)\}$ to $\mathbb{\mu}^{\overline{A}}$.

For the third part, we note that for any
$\Gamma\in\mathbb{B}(\mathbb{S}_{+}^{N})$
\begin{equation}
\label{pr_main2:203}
\frac{1}{N}\sum_{n=1}^{N}\mathbb{P}\left(\widehat{P}_{n}(t)\in\Gamma\right)=
\mathbb{P}\left(\widehat{P}_{q}(t)\in\Gamma\right)
\end{equation}
due to the independence of $q$ from $\{A(t)\}$. Taking the $\limsup$
and noting the non-negativity of the terms, we have for closed
$F$,
{\small
\begin{eqnarray}
\label{pr_main2:204}
\limsup_{t\rightarrow\infty}\mathbb{P}\left(\widehat{P}_{n}(t)\in
F\right) & \leq &
\limsup_{t\rightarrow\infty}\sum_{n=1}^{N}\mathbb{P}\left(\widehat{P}_{n}(t)\in
F\right)\nonumber
\\ & = & N\limsup_{t\rightarrow\infty}\sum_{n=1}^{N}\left[\frac{1}{N}\mathbb{P}\left(\widehat{P}_{n}(t)\in F\right)\right]\nonumber
\\ & = & N\limsup_{t\rightarrow\infty}\mathbb{P}\left(\widehat{P}_{q}(t)\in
F\right)\nonumber  \leq  N\mathbb{\mu}^{\overline{A}}
\end{eqnarray}
}
\noindent
The proof of the second part involves an auxiliary construction and approximation arguments to relate the limit properties of the sequences $\left\{P_{n}(t)\right\}$ to similar processes, where the underlying switching Markov chain is stationary. To this end consider any strictly positive $s\in\mathbb{T}_{+}$. Recall the Markov chains $\left\{p_{n}(t)\right\}$ for $n=1,\cdots,N$ with transition probability matrix $\overline{A}$ and initial state $p_{n}(0)=n$. The corresponding sequence of interacting particle processes $\left\{P_{n}(t)\right\}$ are constructed, for each $n$ as:
\begin{equation}
\label{pr_main2:10}
P_{n}(t+1)=f_{p_{n}(t)}\left(P_{n}(t)\right)
\end{equation}
with initial condition $P_{n}(0)=\widehat{P}(0)$. Let $f_{0}:\mathbb{S}_{+}^{N}\longmapsto\mathbb{S}_{+}^{N}$ denote the Lyapunov operator
\begin{equation}
\label{pr_main2:11}
f_{0}(X)=\mathcal{F}X\mathcal{F}^{T}+\mathcal{Q}
\end{equation}
and note that the following ordering holds:
\begin{equation}
\label{pr_main2:12}
f_{n}(X)\preceq f_{0}(X),~~~\forall n~\mbox{and}~X\in\mathbb{S}_{+}^{N}
\end{equation}
For a given $s>0$ chosen above and for all $n$, define the processes $\left\{P_{n}^{s}(t)\right\}_{t\geq s}$ by
\begin{equation}
\label{pr_main2:13}
P_{n}^{s}(t+1)=f_{p_{n}(t)}\left(P_{n}^{s}(t)\right)
\end{equation}
with deterministic initial value $P_{n}^{s}(s)=f_{0}^{s}\left(\widehat{P}(0)\right)$.
By eqn.~(\ref{pr_main2:12}), for any $s$ tuple $\left(i_{0},i_{1},\cdots,i_{s-1}\right)$ with $i_{r}\in [1,\cdots,N]$ for $r=0,\cdots,s-1$, we note that
\begin{equation}
\label{pr_main2:15}
f_{i_{s-1}}\circ f_{i_{s-2}}\circ\cdots\circ f_{i_{0}}\left(\widehat{P}(0)\right)\preceq f_{0}^{s}\left(\widehat{P}(0)\right)
\end{equation}
and, hence, by the monotonicity of the Riccati operators, we conclude that for all $n$
\begin{equation}
\label{pr_main2:16}
P_{n}(t)\preceq P_{n}^{s}(t),~~~t\geq s
\end{equation}
Also consider a stationary Markov chain $\{q(t)\}_{t\geq s}$ with transition probability $\overline{A}$, i.e., $q(0)$ is uniformly distributed on $V$, and define the process $\left\{Q^{s}(t)\right\}_{t\geq s}$ by
\begin{equation}
\label{pr_main2:17}
Q^{s}(t+1)=f_{q(t)}\left(Q^{s}(t)\right)
\end{equation}
with deterministic initial value
\begin{equation}
\label{pr_main2:18}
Q^{s}(s)=f_{0}^{s}\left(\widehat{P}(0)\right)
\end{equation}
It is to be noted that by Theorem~\ref{conv_aux}, the process $\left\{Q^{s}(t)\right\}$ converges weakly to $\mathbb{\mu}^{\overline{A}}$, i.e.,
\begin{equation}
\label{pr_main2:19}
\lim_{t\rightarrow\infty}d_{P}\left(Q^{s}(t),\mathbb{\mu}^{\overline{A}}\right)=0
\end{equation}
where $d_{P}$ denotes the Prohorov metric. We now set to relate the limit properties of $\left\{P_{n}^{s}(t)\right\}$ to those of $\left\{Q^{s}(t)\right\}$. For $t\geq s$ define the total variation distance between $P_{n}^{s}(t)$ and $Q^{s}(t)$ by
\begin{equation}
\label{pr_main2:20}
d_{v}\left(P_{n}^{s}(t),Q^{s}(t)\right)=\sup_{\Gamma\in\mathcal{B}(\mathbb{S}_{+}^{N})}\left|\mathbb{P}\left(P_{n}^{s}(t)\in\Gamma\right)-\mathbb{P}\left(Q^{s}(t)\in\Gamma\right)\right|
\end{equation}
Since for any $t$, the two sequences considered above assume values in a finite set, we define a set of $(t-s)$ tuples $\Lambda(\Gamma)$ by
\begin{equation}
\label{pr_main2:21}
\Lambda(\Gamma)=\left\{(i_{1},\cdots,i_{t-s})~|~i_{r}\in [1,\cdots,N]~\mbox{for all r and}~f_{i_{t-s}}\circ\cdots\circ f_{i_{1}}\left(f_{0}^{s}(\widehat{P}(0))\right)\right\}
\end{equation}
It is clear that
\begin{eqnarray}
\label{pr_main2:22}
\left\{P_{n}^{s}(t)\in\Gamma\right\}
&\Longleftrightarrow&\left\{(p_{n}(s),\cdots,p_{n}(t-1))\in\Lambda(\Gamma)\right\}\\
\label{pr_main2:23}
\left\{Q^{s}(t)\in\Gamma\right\}&\Longleftrightarrow&\left\{(q(s),\cdots,q(t-1))\in\Lambda(\Gamma)\right\}
\end{eqnarray}
We then have
{\small
\begin{eqnarray}
\label{pr_main2:24}
\mathbb{P}\left(P_{n}^{s}(t)\in\Gamma\right)-\mathbb{P}\left(Q^{s}(t)\in\Gamma\right) & = & \sum_{i_{1}}\mathbb{P}\left(p_{n}(s)=i_{1}\right)\sum_{\left(i_{1},\cdots,i_{t-s}\right)\in\Lambda(\Gamma)}
\prod_{r=1}^{t-s-1}\overline{A}_{i_{r}i_{r+1}}\nonumber \\ & & -\sum_{i_{1}}\mathbb{P}\left(q(s)=i_{1}\right)\sum_{\left(i_{1},\cdots,i_{t-s}\right)\in\Lambda(\Gamma)}
\prod_{r=1}^{t-s-1}\overline{A}_{i_{r}i_{r+1}}\nonumber \\ & = &
\nonumber
\sum_{i_{1}}\left[\mathbb{P}\left(p_{n}(s)=i_{1}\right)-\mathbb{P}\left(q(s)=i_{1}\right)\right]
\sum_{\left(i_{1},\cdots,i_{t-s}\right)\in\Lambda(\Gamma)}\prod_{r=1}^{t-s-1}\overline{A}_{i_{r}i_{r+1}}
\end{eqnarray}
}
and hence
{\small
\begin{eqnarray}
\label{pr_main2:25}
\left|\mathbb{P}\left(P_{n}^{s}(t)\in\Gamma\right)-\mathbb{P}\left(Q^{s}(t)\in\Gamma\right)\right| & \leq & \sum_{i_{1}}\left|\mathbb{P}\left(p_{n}(s)=i_{1}\right)-\mathbb{P}\left(q(s)=i_{1}\right)\right|
\sum_{(i_{1},\cdots,i_{t-s})\in\Lambda(\Gamma)}\prod_{r=1}^{t-s-1}\overline{A}_{i_{r}i_{r+1}}\nonumber \\ & \leq & \sum_{i_{1}}\left|\mathbb{P}\left(p_{n}(s)=i_{1}\right)-\mathbb{P}\left(q(s)=i_{1}\right)\right|\nonumber \\ & \leq & \sum_{i_{1}}d_{v}(p_{n}(s),q(s))\nonumber  \leq Nd_{v}\left(p_{n}(s),q(s)\right)
\end{eqnarray}
}
\noindent
where we have used the fact that
\begin{eqnarray}
\label{pr_main2:26}
\sum_{\left(i_{1},\cdots,i_{t-s}\right)\in\Lambda(\Gamma)}\prod_{r=1}^{t-s-1}\overline{A}_{i_{r}i_{r+1}} & = & \mathbb{P}\left(\left(p_{n}(s+1),\cdots,p_{n}(t-1)\right)=
\left(i_{2},\cdots,i_{t-s}\right)\left|p_{n}(s)=i_{1}\right.\right)\nonumber \leq  1
\end{eqnarray}
We thus obtain
\begin{equation}
\label{pr_main2:27}
d_{v}\left(P_{s}^{n}(t),Q^{s}(t)\right)\leq Nd_{v}\left(p_{n}(s),q(s)\right),~~\forall t\geq s
\end{equation}
It is well known that the finite state Markov chain $\{p_{n}(s)\}$ converges weakly at a geometric rate to the uniform measure, i.e., the measure induced by $q(s)$ for each $s$ and hence in variation. In other words,
\begin{equation}
\label{pr_main2:28}
\lim_{s\rightarrow\infty}d_{v}(p_{n}(s),q(s))=0
\end{equation}
Thus, by eqn.~(\ref{pr_main2:27}), we have
\begin{equation}
\label{pr_main2:29}
\lim_{s\rightarrow\infty}\sup_{t\geq s}d_{v}\left(P_{s}^{n}(t),Q^{s}(t)\right)=0
\end{equation}
and, since convergence in total variation implies weak convergence (\cite{Ethier-Kurtz}), we have
\begin{equation}
\label{pr_main2:30}
\lim_{s\rightarrow\infty}\sup_{t\geq s}d_{P}\left(P_{s}^{n}(t),Q^{s}(t)\right)=0
\end{equation}
Now consider $\varepsilon>0$. Then there exists $s(\varepsilon)$, such that,
\begin{equation}
\label{pr_main2:31}
d_{P}\left(P_{s}^{n}(t),Q^{s}(t)\right)\leq\varepsilon/2,~~~s\geq s(\varepsilon),t\geq s
\end{equation}
Since the sequence $\left\{Q^{s}(t)\right\}$ converges weakly to $\mathbb{\mu}^{\overline{A}}$ for all $s$ (in particular for $s=s(\varepsilon)$,) there exists $t(\varepsilon)\geq s(\varepsilon)$ sufficiently large, such that,
\begin{equation}
\label{pr_main2:32}
d_{P}\left(Q^{s(\varepsilon)}(t),\mathbb{\mu}^{\overline{A}}\right)\leq\varepsilon/2,~~~t\geq t(\varepsilon)
\end{equation}
Then, an application of the triangle inequality for the metric $d_{P}$ leads to
\begin{equation}
\label{pr_main2:33}
d_{P}\left(P_{n}^{s(\varepsilon)}(t),\mathbb{\mu}^{\overline{A}}\right) \leq d_{P}\left(P_{n}^{s(\varepsilon)}(t),Q^{s(\varepsilon)}(t)\right)+d_{P}\left(Q^{s(\varepsilon)}(t),\mathbb{\mu}^{\overline{A}}\right)\leq\varepsilon
\end{equation}
for all $t\geq t(\varepsilon)$. Now, by definition,
\begin{equation}
\label{pr_main2:34}
d_{P}\left(P_{n}^{s(\varepsilon)}(t),\mathbb{\mu}^{\overline{A}}\right)=
\inf\left\{\delta>0\left|\mathbb{P}\left(P_{n}^{s(\varepsilon)}(t)\in F\right)\leq\mathbb{\mu}^{\overline{A}}\left(F^{\delta}\right)+\delta~\mbox{for all closed $F\in\mathbb{S}_{+}^{N}$}\right.\right\}
\end{equation}
where $F^{\delta}$ is defined as
\begin{equation}
\label{pr_main2:35-a}
F^{\delta}=\left\{X\in\mathbb{S}_{+}^{N}\left|\inf_{Y\in F}\left\|X-Y\right\|<\delta\right.\right\}
\end{equation}
Since, by eqn.~(\ref{pr_main2:33}), $d_{P}\left(P_{n}^{s(\varepsilon)}(t),\mathbb{\mu}^{\overline{A}}\right)\leq\varepsilon$ for all $t\geq t(\varepsilon)$, we have, for any closed set $F$,
\begin{equation}
\label{pr_main2:35}
\mathbb{P}\left(P_{n}^{s(\varepsilon)}(t)\in F\right)\leq\mathbb{\mu}^{\overline{A}}\left(F^{\varepsilon}\right)+\varepsilon,~~t\geq t(\varepsilon)
\end{equation}

In addition to $F$ being closed, let us assume that $F$ satisfies monotonicity, i.e., $X\in F$ implies $Y\in F$ for all $Y\succeq X$. By eqn.~(\ref{pr_main2:16}) we have
\begin{equation}
\label{pr_main2:36}
P_{n}(t)\preceq P_{n}^{s(\varepsilon)}(t),~~~t\geq t(\varepsilon)\geq s(\varepsilon)
\end{equation}
and hence
\begin{equation}
\label{pr_main2:37}
\mathbb{P}\left(P_{n}(t)\in F\right)\leq\mathbb{P}\left(P_{n}^{s(\varepsilon)}(t)\in F\right),~~~t\geq t(\varepsilon)
\end{equation}
We then have from eqn.~(\ref{pr_main2:35}) for all $t\geq t(\varepsilon)$
\begin{equation}
\label{pr_main2:38}
\mathbb{P}\left(P_{n}(t)\in F\right) \leq\mathbb{\mu}^{\overline{A}} \left(F^{\varepsilon}\right)+\varepsilon
\end{equation}
Taking the limit as $t\rightarrow\infty$, we have
\begin{equation}
\label{pr_main2:39}
\limsup_{t\rightarrow\infty}\mathbb{P}\left(P_{n}(t)\in F\right)\leq\mathbb{\mu}^{\overline{A}}\left(F^{\varepsilon}\right)+\varepsilon
\end{equation}
The L.H.S. above is now independent of $t$ and, hence, $\varepsilon$ through $t(\varepsilon)$. Since the above holds for arbitrary $\varepsilon>0$, moving to the limit as $\varepsilon\rightarrow 0$ yields
\begin{equation}
\label{pr_main2:40}
\limsup_{t\rightarrow\infty}\mathbb{P}\left(P_{n}(t)\in F\right)\leq\lim_{\varepsilon\rightarrow 0}\mathbb{\mu}^{\overline{A}}\left(F^{\varepsilon}\right)=\mathbb{\mu}^{\overline{A}}\left(F\right)
\end{equation}
The last step follows from the continuity of the probability measure $\mathbb{\mu}^{\overline{A}}$ and the fact that
\begin{equation}
\label{pr_main2:41}
\bigcap_{\varepsilon>0}F^{\varepsilon}=F
\end{equation}
for closed $F$. This establishes the result for general order preserving $F$. The result for sets of the form $\left\{\left.X\in\mathbb{S}_{+}^{N}\right|X\succeq\alpha I\right\}$ or $\left\{\left.X\in\mathbb{S}_{+}^{N}\right|\|X\|\geq\alpha\right\}$ for $\alpha>0$ follow, as they satisfy the general hypothesis on $F$.
\end{proof}

\section{Concluding remarks}
\label{conclusion} The paper develops the gossip interactive Kalman filter (GIKF) for distributed Kalman filtering in sensor networks, when observation sampling and inter-sensor communication occur at the same time scale. Inter-sensor collaboration is achieved by intermittent exchange of filtering states. A traveling particle interpretation of the filtering states leads to a random dynamical system (RDS) formulation of the sequence of conditional error covariances. Under a weak detectability assumption, the estimation error process at each sensor stays stochastically bounded (irrespective of the instability in signal dynamics,) provided the network satisfies some weak connectivity conditions. Also, the network achieves weak consensus, i.e., the conditional error covariance (or the pathwise filtering error) at a randomly selected sensor converges in distribution to a unique invariant measure $\mathbb{\mu}^{\overline{A}}$. The invariant measure $\mathbb{\mu}^{\overline{A}}$ depends on the network connectivity process (the MAC protocol) through the mean  $\overline{A}$ of the random adjacency matrix~$A$.

The characterization of the invariant measure $\mathbb{\mu}^{\overline{A}}$ as a functional of the matrix is  interesting to study the sensitivity of the mapping, $\overline{A}\longrightarrow\mathbb{\mu}^{\overline{A}}$. This would lead to understanding the robustness of the above filtering approach to perturbations in the communication policy, i.e., whether a small change in the MAC protocol (a perturbation of $\overline{A}$) leads to a negligible change of $\mathbb{\mu}^{\overline{A}}$, or the filtering performance changes dramatically. Exploring such comparison principles for the mapping would lead to understanding the more complicated problem of characterizing the invariant measure $\mathbb{\mu}^{\overline{A}}$. Such a characterization, in general, is difficult as there seems to be no direct way of obtaining a functional mapping $\overline{A}$ to $\mathbb{\mu}^{\overline{A}}$. In fact, a much simpler situation (Kalman filtering with intermittent observations) involving a single sensor with observation packet losses demands the machinery of moderate deviations (\cite{Riccati-moddev}) and large random matrix theory (\cite{vh08}) for a characterization of the invariant measures.


\appendices

\renewcommand{\baselinestretch}{1.16}

\section{Random Dynamical Systems: Facts and Results}
\label{RDS}
{\small
We start by defining a random dynamical system (RDS). In the sequel, we follow the notation in~\cite{ArnoldChueshov,Chueshov}.
\begin{definition}[RDS]\label{defn_RDS} A RDS with (one-sided) time
$\mathbb{T}_{+}$ and state space $\mathcal{X}$ is a pair
$(\theta,\varphi)$ with the following properties:
\begin{itemize}
\item[\textbf{A)}] A metric dynamical system
$\theta=\left(\Omega,\mathcal{F},\mathbb{P},\{\theta_{t},t\in\mathbb{T}\}\right)$
with two-sided time $\mathbb{T}$, i.e., a probability space
$(\Omega,\mathcal{F},\mathbb{P})$ with a family of transformations
$\left\{\theta_{t}:\Omega\longmapsto\Omega\right\}_{t\in\mathbb{T}}$ such
that
\begin{itemize}
\item[\textbf{A.1)}]$\theta_{0}=\mbox{id}_{\Omega},~~~\theta_{t}\circ\theta_{s}=\theta_{t+s},~~~\forall
t,s\in\mathbb{T}$
\item[\textbf{A.2)}]$(t,\omega)\longmapsto\theta_{t}\omega$ is
measurable.
\item[\textbf{A.3)}]$\theta_{t}\mathbb{P}=\mathbb{P}~~\forall
t\in\mathbb{T}$, i.e., $\mathbb{P}\left(\theta_{t}B\right)=\mathbb{P}$ for all $B\in\mathcal{F}$ and all $t\in\mathbb{T}$.
\end{itemize}
\item[\textbf{B)}] A cocycle $\varphi$ over $\theta$ of continuous
mappings of $\mathcal{X}$ with time $\mathbb{T}_{+}$, i.e., a
measurable mapping
\begin{equation}
\label{def_RDS} \varphi:\mathbb{T}_{+}\times\Omega\times
\mathcal{X},~~(t,\omega,X)\longmapsto\varphi(t,\omega,X)
\end{equation}
such that
\begin{itemize}
\item[\textbf{B.1)}] The mapping
$X\longmapsto\varphi(t,\omega,X)\equiv\varphi(t,\omega)X$ is
continuous in $X$ for every $t\in\mathbb{T}_{+}$ and
$\omega\in\Omega$.
\item[\textbf{B.2)}] The mappings
$\varphi(t,\omega)\doteq\varphi(t,\omega,\cdot)$ satisfy the
cocycle property:
\begin{equation}
\label{def_RDS1}
\varphi(0,\omega)=\mbox{id}_{\mathcal{X}},~~\varphi(t+s,\omega)=\varphi\left(t,\theta_{s}\omega\right)
\circ\varphi(s,\omega)
\end{equation}
for all $t,s\in\mathbb{T}_{+}$ and $\omega\in\Omega$.
\end{itemize}
\end{itemize}
\end{definition}
Although we consider in this paper discrete time RDS, the general
notion of RDS, as defined in~\cite{Chueshov}, applies equally well
to dynamical systems with continuous time. In the above
definition, the randomness is captured by the probability space
$(\Omega,\mathcal{F},\mathbb{P})$ and iterates indexed by $\omega$
indicates pathwise construction. For example, if $X_{0}$ is the
deterministic initial state of the system of interest at time
$t=0$, the random state at time $t\in\mathbb{T}_{+}$ is given by
\begin{equation}
\label{def_RDS2} X_{t}(\omega)=\varphi\left(t,\omega,X_{0}\right)
\end{equation}
The measurability assumptions in the definition above, guarantee
that the random state $X_{t}$ is a well-defined random variable.
Also, note that the iterates are defined for non-negative
(one-sided) time, however, the family of transformations
$\left\{\theta_{t}\right\}$ is two-sided, which is purely for technical
convenience, as will be seen later.

\textbf{Some results from RDS theory}
We summarize terminology and notions used in the RDS literature (see~\cite{ArnoldChueshovbook,Chueshov} for details.)

Consider a generic RDS $(\theta,\varphi)$ with state space
$\mathcal{X}$ as in Definition~\ref{defn_RDS}. In the following we
assume that $\mathcal{X}$ is a non-empty subset of a real Banach
space $V$ with a closed, convex, solid, normal (w.r.t.~the Banach
space norm,) minihedral cone $V_{+}$. We denote by $\preceq$ the
partial order induced by $V_{+}$ in $\mathcal{X}$ and $<<$ denotes
the corresponding strong order. Although the development that
follows may hold for arbitrary $\mathcal{X}\subset V$, in the
sequel we assume $\mathcal{X}=V_{+}$ (which is true for the RDS
$\left(\theta^{R},\varphi^{R}\right)$ modeling the RARE.)

\begin{definition}[Order-Preserving RDS]
\label{Order} An RDS $(\theta,\varphi)$ with state space $V_{+}$
is called order-preserving if
\begin{equation}
X\preceq
Y~\Longrightarrow~\varphi(t,\omega,X)\preceq\varphi(t,\omega,Y),~~\forall
t\in\mathbb{T}_{+}~,\omega\in\Omega~,X,Y\in V_{+}
\end{equation}
\end{definition}
\begin{definition}[Sublinearity]\label{sublinearity} An order-preserving RDS $(\theta,\varphi)$ with state space $V_{+}$ is
called sublinear if for every $X\in V_{+}$ and $\lambda\in (0,1)$
we have
\begin{equation}
\label{prop_RDS}
\lambda\varphi(t,\omega,X)\preceq\varphi(t,\omega,\lambda
X),~~\forall t>0,~\omega\in\Omega
\end{equation}
The RDS is said to be strictly sublinear if strict inequality in
eqn.~(\ref{prop_RDS}) holds for $X\in\mbox{int}V_{+}$, i.e. for
$X\in\mbox{int}V_{+}$,
\begin{equation}
\label{prop_RDS1}
\lambda\varphi(t,\omega,X)\prec\varphi(t,\omega,\lambda
X),~~\forall t>0,~\omega\in\Omega
\end{equation}
and strongly sublinear if in addition to eqn.~(\ref{prop_RDS}), we
have
\begin{equation}
\label{prop_RDS2}
\lambda\varphi(t,\omega,X)\ll\varphi(t,\omega,\lambda X),~~\forall
t>0,~\omega\in\Omega,~X\in\mbox{int}V_{+}
\end{equation}
\end{definition}
\begin{definition}[Equilibrium]\label{equilibrium} A random
variable $u:\Omega\longmapsto V_{+}$ is called an equilibrium
(fixed point, stationary solution) of the RDS $(\theta,\varphi)$
if it is invariant under $\varphi$, i.e.,
\begin{equation}
\label{prop_RDS3}
\varphi\left(t,\omega,u(\omega)\right)=u\left(\theta_{t}\omega\right),~~\forall
t\in\mathbb{T}_{+}~,\omega\in\Omega
\end{equation}
In case, eqn.~(\ref{prop_RDS3}) holds for all $\omega\in\Omega$,
except on a set of $\mathbb{P}$ measure zero, we call $u$ an
almost equilibrium.
\end{definition}
Since, the transformations $\left\{\theta_{t}\right\}$ are
measure-preserving, i.e.,
$\theta_{t}\mathbb{P}=\mathbb{P},~~\forall t$, we have
\begin{equation}
\label{prop_RDS4} u\left(\theta_{t}\omega\right)\ndtstile{}{d}
u(\omega),~~\forall t
\end{equation}
Thus eqn.~(\ref{prop_RDS3}), in particular, implies that, for an
almost equilibrium $u$, the sequence of iterates
$\left\{\varphi\left(t,\omega,u(\omega)\right)\right\}_{t\in\mathbb{T}_{+}}$
have the same distribution, which is the distribution of $u$.
\begin{definition}[Part]\label{part} The equivalence classes in
$V_{+}$ under the equivalence relation defined by $X\sim Y$ if
there exists $\alpha_{0}\geq 1$ such that $\alpha_{0}^{-1}X\preceq
Y\preceq\alpha_{0} X$ are called parts of $V_{+}$.

We call the part $C_{v}$ generated by a random variable
$v:\Omega\longmapsto V_{+}$ as the collection of random variables
$u:\Omega\longmapsto V_{+}$ such that there exists deterministic
$\alpha_{u}\geq 1$ with
\begin{equation}
\label{part1} \alpha_{u}^{-1}v(\omega)\preceq
u(\omega)\preceq\alpha_{u}v(\omega),~~\forall\omega\in\Omega
\end{equation}
\end{definition}
\begin{definition}[Orbit]\label{orbit} For a random variable $u:\Omega\longmapsto
V_{+}$ we define the \emph{forward} orbit $\eta^{f}_{u}(\omega)$
emanating from $u(\omega)$ as the random set
$\left\{\varphi\left(t,\omega,u(\omega)\right)\right\}_{t\in\mathbb{T}_{+}}$.
The forward orbit gives the sequence of iterates of the RDS
starting at $u$.

Although $\eta^{f}_{u}$ is the object of practical interest,
for technical convenience (will be seen later,) we also define the
\emph{pull-back} orbit $\eta^{b}_{u}(\omega)$ emanating from $u$
as the random set
$\left\{\varphi\left(t,\theta_{-t}\omega,u(\theta_{-t}\omega)\right)\right\}_{t\in\mathbb{T}_{+}}$.
\end{definition}
The reason for defining the pull-back orbit is that it is
comparatively convenient to establish asymptotic properties for
$\eta^{b}_{u}$. However, analyzing $\eta_{u}^{b}$ leads to
understanding asymptotic distributional properties for
$\eta^{f}_{u}$, because the random sequences
$\left\{\varphi\left(t,\omega,u(\omega)\right)\right\}_{t\in\mathbb{T}_{+}}$
and
$\left\{\varphi\left(t,\theta_{-t}\omega,u(\theta_{-t}\omega)\right)\right\}_{t\in\mathbb{T}_{+}}$
are equivalent in distribution. In other words,
\begin{equation}
\label{orbit1}
\varphi\left(t,\omega,u(\omega)\right)\ndtstile{}{d}\varphi\left(t,\theta_{-t}\omega,u(\theta_{-t}\omega)\right),~~\forall
t\in\mathbb{T}_{+}
\end{equation}
This follows from the fact that
$\theta_{t}\mathbb{P}=\mathbb{P},~~\forall t\in\mathbb{T}$. Thus,
in particular, we have the following assertion.
\begin{lemma}
\label{orbit_lemma} Let the sequence
$\left\{\varphi\left(t,\theta_{-t}\omega,u\left(\theta_{-t}\omega\right)\right)\right\}_{t\in\mathbb{T}_{+}}$
converge in distribution to a measure $\mu$ on $V_{+}$, where
$u:\Omega\longmapsto V_{+}$ is a random variable. Then the
sequence
$\left\{\varphi\left(t,\omega,u(\omega)\right)\right\}_{t\in\mathbb{T}_{+}}$
also converges in distribution to the measure $\mu$.
\end{lemma}
We now introduce notions of boundedness of RDS, which will be
used in the sequel.
\begin{definition}[Boundedness]\label{boundedness} Let
$a:\Omega\longmapsto V_{+}$ be a random variable. The pull-back
orbit $\eta_{a}^{b}(\omega)$ emanating from $a$ is said to be
bounded on $U\in\mathcal{F}$ is there exists a random variable $C$
on $U$ such that
\begin{equation}
\label{boundedness1}
\left\|\varphi\left(t,\theta_{-t}\omega,a\left(\theta_{-t}\omega\right)\right)\right\|\leq
C(\omega),~~\forall t\in\mathbb{T}_{+},~\omega\in U
\end{equation}
\end{definition}
\begin{definition}[Conditionally Compact RDS]\label{cond_comp} An
RDS $(\theta,\varphi)$ in $V_{+}$ is said to be conditionally
compact if for any $U\in\mathcal{F}$ and pull-back orbit
$\eta_{a}^{b}(\omega)$ which is bounded on $U$ there exists a
family of compact sets $\left\{K(\omega)\right\}_{\omega\in U}$ such that
\begin{equation}
\label{cond_comp1}
\lim_{t\rightarrow\infty}\mbox{dist}\left(\varphi\left(t,\theta_{-t}\omega,a\left(\theta_{-t}\omega\right)
\right),K(\omega)\right)=0,~~\omega\in U
\end{equation}
\end{definition}
It is to be noted that conditionally compact is a topological
property of the space $V_{+}$.
%


We now state a limit set dichotomy result for a class of
sublinear, order-preserving RDS.
\begin{theorem}[Corollary 4.3.1. in~\cite{Chueshov}]\label{LSD}
Let $V$ be a separable Banach space with a normal solid cone
$V_{+}$. Assume that $(\theta,\varphi)$ is a strongly sublinear
conditionally compact order-preserving RDS over an ergodic metric
dynamical system $\theta$. Suppose that $\varphi(t,\omega,0)\gg 0$
for all $t>0$ and $\omega\in\Omega$. Then precisely one of the
following applies:
\begin{itemize}
\item[\textbf{(a)}] For any $X\in V_{+}$ we have
\begin{equation}
\label{LSD1}
\mathbb{P}\left(\lim_{t\rightarrow\infty}\left\|\varphi\left(t,\theta_{-t}\omega,X\right)\right\|=\infty\right)=1
\end{equation}
\item[\textbf{(b)}] There exists a unique almost equilibrium
$u(\omega)\gg 0$ defined on a $\theta$-invariant set\footnote{A
set $A\in\mathcal{F}$ is called $\theta$-invariant if
$\theta_{t}A=A$ for all $t\in\mathbb{T}$.}
$\Omega^{\ast}\in\mathcal{F}$ with
$\mathbb{P}\left(\Omega^{\ast}\right)=1$ such that for any random
variable $v(\omega)$ possessing the property $0\preceq
v(\omega)\preceq\alpha_{0} u(\omega)$ for all $\omega\in\Omega^{\ast}$
and deterministic $\alpha_{0}>0$, the following holds:
\begin{equation}
\label{LSD2}
\lim_{t\rightarrow\infty}\varphi\left(t,\theta_{-t}\omega,v(\theta_{-t}\omega)\right)=u(\omega),~~\omega\in\Omega^{\ast}
\end{equation}
\end{itemize}
\end{theorem}
}

\section{Proofs in Section~\ref{RDS_auxproof}}
\label{RDS_123}
{\small

\textbf{Proof of Lemma~\ref{dist_obsRiccati}}
The proof is obtained by constructing an approximate filter with suboptimal performance, and then bounding its error by using the rank condition on the Grammian $G_{w_{0}}$. We detail such a construction now.

Consider $\ell$ steps $t=1,\cdots,\ell$ of the
 linear time-varying signal/observation model given by~\ref{sys_model}
 and~\ref{obs_n}, where in~\ref{obs_n} we index the observation matrix
 as~$\mathcal{C}_{n_{t}}$ where $n(t)$ indicates the current state of the walk~$w_0$.
 The signal vector $\mathbf{x}_{t}\in\mathbb{R}^{M}$ with initial state $\mathbf{x}_{1}$ is a Gaussian random variable with known mean $\overline{\mathbf{x}}_{1}$ and variance $X\in\mathbb{S}_{+}^{N}$. The system noise process $\{\mathbf{w}_{t}\}$ is uncorrelated zero mean Gaussian with covariance $\mathcal{Q}$. The observation noise process $\{\mathbf{v}_{t}\}_{t=1}^{\ell}$ is uncorrelated zero mean Gaussian with time varying error covariance $\mathcal{R}_{n_{t}}$ and independent of the initial signal state and the system noise process. By the above construction, the optimal estimate of the signal state $\mathbf{x}_{t}$ at time $t$, based on observations till that time, is given by the Kalman filter initialized with $X$ as the predicted conditional error covariance at time $t=1$. In other words, the optimal m.m.s.e.~state estimator (predictor form)
\begin{equation}
\label{unif6}\widehat{\mathbf{x}}_{w_{0}}(t)=\mathbb{E}\left[\mathbf{x}_{t}\left|\{\mathbf{y}_{s}\}_{1\leq s<t}\right.\right]
\end{equation}
of $\mathbf{x}_{t}$ based on observations $\left\{\mathbf{y}_{s}\right\}_{1\leq s<t}$ for $1\leq t\leq\ell +1$ can be recursively constructed through the Kalman filter and the corresponding predicted conditional error covariance sequence $\left\{P_{w_{0}}(t)\right\}_{1\leq t\leq\ell+1}$ satisfies the recursion:
\begin{equation}
\label{unif7} P_{w_{0}}(t+1)=\mathcal{F}P_{w_{0}}(t)\mathcal{F}^{T}+\mathcal{Q}-\mathcal{F}P_{w_{0}}(t)\mathcal{C}_{n_{t}}^{T}
\left(\mathcal{C}_{n_{t}}P_{w_{0}}(t)\mathcal{C}_{n_{t}}^{T}+\mathcal{R}_{n_{t}}\right)^{-1}\mathcal{C}_{n_{t}}P_{w_{0}}(t)\mathcal{F}^{T}
\end{equation}
with initial condition $P_{w_{0}}(1)=X$. We then have
\begin{equation}
\label{unif8}P_{w_{0}}(\ell+1)=f_{n_{\ell}}\circ\cdots\circ f_{n_{1}}(X)
\end{equation}
the R.H.S.~being the desired functional form in eqn.~(\ref{unif3}), i.e., $P_{w_{0}}(\ell+1)=g_{w_{0}}(X)$. Since for a Kalman filter with deterministic system/observation matrices, the conditional error covariance is equal to the unconditional one and the fact that the Kalman filter minimizes any positive definite form of the estimation error, for a generic estimator $\widehat{\mathbf{h}}$ of $\mathbf{x}_{\ell+1}$ based on $\{\mathbf{y}_{s}\}_{1\leq s\leq\ell}$ we have
\begin{equation}
\label{unif100}
P_{w_{0}}(\ell+1)\preceq\mathbb{E}\left[\left(\mathbf{x}_{\ell+1}-\widehat{\mathbf{h}}\right)
\left(\mathbf{x}_{\ell+1}-\widehat{\mathbf{h}}\right)^{T}\right]
\end{equation}
where $\preceq$ refers to the partial order on $\mathbb{S}_{+}^{N}$.
In order to upper bound the functional $g_{w_{0}}$, we now construct a suboptimal state estimator with a guaranteed estimation performance. To this end, define the modified Grammian
\begin{equation}
\label{unif10}\widetilde{G}_{w_{0}}=\sum_{t=1}^{\ell}\left(\mathcal{F}^{t-1}\right)^{T}
\mathcal{C}_{n_{t}}^{T}\mathcal{R}_{n_{t}}^{-1}\mathcal{C}_{n_{t}}\mathcal{F}^{t-1}
\end{equation}
We note that $\widetilde{G}_{w_{0}}$ is invertible by the invertibility of $G_{w_{0}}$ and the noise covariances $\mathcal{R}_{n_{t}}$. Define the suboptimal estimator of $\mathbf{x}_{\ell+1}$ by:
\begin{equation}
\label{unif11}\overline{\mathbf{x}}_{w_{0}}(\ell+1)=\mathcal{F}^{\ell}\widetilde{G}_{w_{0}}^{-1}
\sum_{t=1}^{\ell}\left(\mathcal{F}^{t-1}\right)^{T}\mathcal{C}_{n_{t}}\mathcal{R}_{n_{t}}^{-1}\mathbf{y}_{t}
\end{equation}
based on observations $\left\{\mathbf{y}_{s}\right\}_{1\leq s\leq\ell}$. Using the fact, that,
\begin{equation}
\label{unif12}\mathbf{x}_{t}=\mathcal{F}^{t-1}\mathbf{x}_{1}+
\sum_{s=1}^{t-1}\mathcal{F}^{t-1-s}\mathbf{w}_{s},~~~1\leq t\leq\ell +1
\end{equation}
we have from eqn.~(\ref{unif11})
\begin{eqnarray}
\label{unif13}\overline{\mathbf{x}}_{w_{0}}(\ell+1)&=&\mathcal{F}^{\ell}\widetilde{G}_{w_{0}}^{-1}
\sum_{t=1}^{\ell}\left(\mathcal{F}^{t-1}\right)^{T}\mathcal{C}_{n_{t}}\mathcal{R}_{n_{t}}^{-1}\mathcal{C}_{n_{t}}
\mathcal{F}^{t-1}\mathbf{x}_{1}+
\mathcal{F}^{\ell}\widetilde{G}_{w_{0}}^{-1}\sum_{t=1}^{\ell}\left(\mathcal{F}^{t-1}\right)^{T}
\mathcal{C}_{n_{t}}\mathcal{R}_{n_{t}}^{-1}\mathcal{C}_{n_{t}}\sum_{s=1}^{t-1}\mathcal{F}^{t-1-s}\mathbf{w}_{s}
\nonumber \\ & & +\mathcal{F}^{\ell}\widetilde{G}_{w_{0}}^{-1}\sum_{t=1}^{\ell}
\left(\mathcal{F}^{t-1}\right)^{T}\mathcal{C}_{n_{t}}\mathcal{R}_{n_{t}}^{-1}\mathbf{v}_{n_{t}}
\nonumber \\
 & = & \mathcal{F}^{\ell}\mathbf{x}_{1}+
\mathcal{F}^{\ell}\widetilde{G}_{w_{0}}^{-1}\sum_{t=1}^{\ell}\left(\mathcal{F}^{t-1}\right)^{T}
\mathcal{C}_{n_{t}}\mathcal{R}_{n_{t}}^{-1}\mathcal{C}_{n_{t}}\sum_{s=1}^{t-1}\mathcal{F}^{t-1-s}\mathbf{w}_{s}
\nonumber 
+\mathcal{F}^{\ell}\widetilde{G}_{w_{0}}^{-1}\sum_{t=1}^{\ell}\left(\mathcal{F}^{t-1}\right)^{T}
\mathcal{C}_{n_{t}}\mathcal{R}_{n_{t}}^{-1}\mathbf{v}_{n_{t}}
\end{eqnarray}
The filtering error is then given by
\begin{eqnarray}
\label{unif14}
\mathbf{e}_{w_{0}}(\ell+1)& = & \mathbf{x}_{\ell+1}-\overline{\mathbf{x}}_{w_{0}}(\ell+1)
\nonumber \\
 & = & -\mathcal{F}^{\ell}\widetilde{G}_{w_{0}}^{-1}\sum_{t=1}^{\ell}\left(\mathcal{F}^{t-1}\right)^{T}
 \mathcal{C}_{n_{t}}\mathcal{R}_{n_{t}}^{-1}\mathcal{C}_{n_{t}}
 \sum_{s=1}^{t-1}\mathcal{F}^{t-1-s}\mathbf{w}_{s}-\mathcal{F}^{\ell}\widetilde{G}_{w_{0}}^{-1}
 \sum_{t=1}^{\ell}\left(\mathcal{F}^{t-1}\right)^{T}\mathcal{C}_{n_{t}}\mathcal{R}_{n_{t}}^{-1}\mathbf{v}_{n_{t}}
\end{eqnarray}
We note that the error above is independent of the initial state $\mathbf{x}_{1}$ (and hence the covariance $X$) of the system, and the mean square boundedness of the process noise $\{\mathbf{w}_{t}\}$ and observation noise $\{\mathbf{v}_{t}\}$ imply the existence of a constant $\alpha_{0}>0$, such that,
\begin{equation}
\label{unif15}\mathbb{E}\left[\mathbf{e}_{w_{0}}(\ell+1)\mathbf{e}_{w_{0}}(\ell+1)^{T}\right]\preceq\alpha_{0} I
\end{equation}
The Lemma then follows by the optimality of the Kalman filter, as stated in eqn.~(\ref{unif100}).

\textbf{Proof of Lemma~\ref{stoch_boundedness}}
 In case $\mathcal{F}$ is stable, the claim is obvious, as the suboptimal estimate of 0 at each sensor for all time is stochastically bounded. So, in the sequel we assume $\mathcal{F}$ is unstable.

The proof is somewhat technical and mainly uses the uniform boundedness of the composition of Riccati operators in Lemma~\ref{dist_obsRiccati} and the ergodicity of the underlying switching Markov chain $\left\{\widetilde{p}(t)\right\}_{t\in\mathbb{T}_{+}}$. From Lemma~\ref{dist_obsRiccati} it follows that a successive application of $\ell$ Riccati maps (in the composition order $f_{n_{l}}\circ\cdots\circ f_{n_{1}}$) reduces the iterate in the conic interval $[0,\alpha_{0}I]$ irrespective of its initial value. The approach is to relate the probability of large exceedance of $\widetilde{P}_{t}$ to the hitting time statistics of a modified Markov chain. We detail it below.

First, we note that the regularity of the distributions of $\widetilde{P}(t)$ for every $t$, implies that it suffices to show
\begin{equation}
\label{sb200}
\lim_{J\rightarrow\infty}\sup_{t\geq t_{0}}\mathbb{P}\left(\left\|\widetilde{P}(t)\right\|>J\right)=0
\end{equation}
for some arbitrarily large $t_{0}\in\mathbb{T}_{+}$.
For every $n$, the Riccati update is upper bounded by the Lyapunov operator, i.e.,
\begin{equation}
\label{sb2}f_{n}(X)\preceq \mathcal{F}X\mathcal{F}^{T}+\mathcal{Q},~~~\forall X\in\mathbb{S}_{+}^{N}
\end{equation}
For sufficiently large $J>0$, define
\begin{equation}
\label{sb3}k(J)=\max_{k}\left\{k\in\mathbb{T}_{+}\left|\alpha^{2k}\alpha_{0}
+\frac{\alpha^{2k}-1}{\alpha^{2}-1}\|\mathcal{Q}\|\leq J\right.\right\}
\end{equation}
where $\alpha=\|\mathcal{F}\|$. Since $\mathcal{F}$ is unstable ($\alpha>1$), we note that $k(J)\rightarrow\infty$ as $J\rightarrow\infty$.

We introduce additional notation here. For integers $t_{0},t_{1}\geq\ell$, the phrase ``there exists a $(n_{1},n_{2},\cdots,n_{\ell})$ cycle in the interval $[t_{0},t_{1}]$'' indicates the existence of an integer $t_{0}\leq\acute{t}\leq t_{1}$, such that,
\begin{equation}
\label{sb4}\widetilde{p}(\acute{t}-\ell+s)=n_{s},~~~1\leq s\leq\ell
\end{equation}
where $\left\{\widetilde{p}(t)\right\}_{t\in\mathbb{T}_{+}}$ is the switching Markov chain.

We now make the following claim for relating the probabilities of interest for sufficiently large $t$:
\begin{equation}
\label{sb5}
\mathbb{P}\left(\left\|\widetilde{P}(t)\right\|>J\right)\leq\mathbb{P}\left(\mbox{no $(n_{1},n_{2},\cdots,n_{\ell})$ exists in $[t-k(J),t]$}\right)
\end{equation}
Indeed, assume on the contrary that a $(n_{1},\cdots,n_{\ell})$ cycle exists in the interval $[t-k(J),t]$. Then there exists $\acute{t}\in[t-k(J),t]$, such that,
\begin{equation}
\label{sb6}\widetilde{p}(\acute{t}-\ell+s)=n_{s},~~~1\leq s\leq\ell
\end{equation}
This implies
\begin{equation}
\label{sb7}
\widetilde{P}\left(\acute{t}\right)=f_{n_{\ell}}\circ\cdots\circ f_{n_{1}}\left(\widetilde{P}\left(\acute{t}-\ell+1\right)\right)
\end{equation}
and hence by Lemma~\ref{dist_obsRiccati}
\begin{equation}
\label{sb8} \widetilde{P}\left(\acute{t}\right)\preceq\alpha_{0}I
\end{equation}
which holds irrespective of the value of $\widetilde{P}\left(\acute{t}-\ell+1\right)$. By eqn.~(\ref{sb2}) we note that
\begin{equation}
\label{sb9}\widetilde{P}(s)\preceq \mathcal{F}\widetilde{P}(s-1)\mathcal{F}^{T}+\mathcal{Q},~~~\forall s
\end{equation}
Continuing the recursion and noting $\widetilde{P}\left(\acute{t}\right)\preceq\alpha_{0}I$
\begin{eqnarray}
\label{sb10}\left\|\widetilde{P}(t)\right\| & \leq & \alpha^{2\left(t-\acute{t}\right)}\left\|\widetilde{P}\right\|+
\frac{\alpha^{2\left(t-\acute{t}\right)}-1}{\alpha^{2}-1}\|\mathcal{Q}\|\nonumber 
=\alpha^{2\left(t-\acute{t}\right)}\alpha_{0}+\frac{\alpha^{2\left(t-\acute{t}\right)}-1}{\alpha^{2}-1}\|\mathcal{Q}\|
\end{eqnarray}
Since $\left(t-\acute{t}\right)\leq k(J)$, it follows from the above
\begin{eqnarray}
\label{sb11}
\left\|\widetilde{P}(t)\right\| & \leq & \alpha^{2\left(t-\acute{t}\right)}\alpha_{0}+\frac{\alpha^{2\left(t-\acute{t}\right)}-1}{\alpha^{2}-1}\|\mathcal{Q}\|\nonumber
 \leq\alpha^{2k(J)}\alpha_{0}+\frac{\alpha^{2k(J)}-1}{\alpha^{2}-1}\|\mathcal{Q}\|\nonumber
 \leq J
\end{eqnarray}
where the last step follows from the definition of $k(J)$ (eqn.~(\ref{sb3})).
We thus note that the existence of a $\left(n_{1},\cdots,n_{\ell}\right)$ cycle in $[t-k(J),t]$ implies $\left\|\widetilde{P}(t)\right\|\leq J$. i.e., we have the event inclusion:
\begin{equation}
\label{sb12}
\left\{\mbox{there exists a $\left(n_{1},\cdots,n_{\ell}\right)$ cycle in $[t-k(J),t]$}\right\}\subset\left\{\left\|\widetilde{P}(t)\right\|\leq J\right\}
\end{equation}
The claim in eqn.~(\ref{sb5}) follows. Thus estimating the probability on the L.H.S. of eqn.~(\ref{sb5}) reduces to estimating the probability of a $\left(n_{1},\cdots,n_{\ell}\right)$ cycle in $[t-k(J),t]$. To this end we construct another Markov chain $\{z(t)\}_{t\geq\ell}$. The state space $\mathcal{Z}$ is a subset of $V^{\ell}$ given by:
\begin{equation}
\label{sb13}\mathcal{Z}=\left\{z=\left(i_{1},i_{2},\cdots,i_{\ell}\right)~|~\overline{A}_{i_{j},i_{j+1}}>0,~~~1\leq j<\ell\right\}
\end{equation}
The dynamics of the Markov chain $\{z(t)\}_{t\geq \ell}$ is given in terms of the Markov chain $\left\{\widetilde{p}(t)\right\}_{t\in\mathbb{T}_{+}}$ as follows:
\begin{equation}
\label{sb14} z(t)=(\widetilde{p}(t-\ell+1),\widetilde{p}(t-\ell+2),\cdots,\widetilde{p}(t))
\end{equation}
From the dynamics of $\left\{\widetilde{p}(t)\right\}_{t\in\mathbb{T}_{+}}$ it follows that $\{z(t)\}_{t\geq \ell}$ is a Markov chain with transition probability $\overline{A}_{nl}$ between allowable states $\left(i_{1},i_{2},\cdots,i_{\ell-1},n\right)$ and $\left(i_{2},\cdots,\i_{\ell},n,l\right)$. With state space $\mathcal{Z}$, the Markov chain $\{z(t)\}$ inherits irreducibility and aperiodicity from that of $\left\{\widetilde{p}(t)\right\}$. Also, $\{z(t)\}$ is stationary from the stationarity of $\left\{\widetilde{p}(t)\right\}$ with invariant distribution:
\begin{equation}
\label{sb15}
\mathbb{P}\left(z(t)=\left(i_{1},i_{2},\cdots,i_{\ell}\right)\right)=
\frac{1}{N}\prod_{j=1}^{\ell-1}\overline{A}_{j,j+1},~~\left(i_{1},i_{2},\cdots,i_{\ell}\right)\in\mathcal{Z},~t\geq \ell,~t\in\mathbb{T}_{+}
\end{equation}
Denote the hitting time $\tau_{0}$ of $\{z(t)\}$ to the state $\left(n_{1},\cdots,n_{\ell}\right)$ by:
\begin{equation}
\label{sb16}
\tau_{0}=\min\left\{t>\ell\left|z(t)=\left(n_{1},\cdots,n_{\ell}\right)\right.\right\}
\end{equation}
and for all $z\in\mathcal{Z}$ define
\begin{equation}
\label{sb17}
\mathbb{P}_{z}\left(\tau_{0}>s\right)=\mathbb{P}\left(\tau_{0}>s~|~z(\ell)=z\right)
\end{equation}
Also, for each $t\geq\ell$ and $J$ sufficiently large, define the stopping times
\begin{equation}
\label{sb18}
\tau_{t}^{J}=\min\left\{t\geq t-k(J)\left|z(t)=\left(n_{1},\cdots,n_{\ell}\right)\right.\right\}
\end{equation}
From the Markov property it then follows
\begin{equation}
\label{sb19}
\mathbb{P}\left(\tau_{t}^{J}>t~|~z(t-k(J)-1)=z\right)=\mathbb{P}_{z}\left(\tau_{0}>k(J)+1\right)
\end{equation}
It then follows successively
\begin{eqnarray}
\label{sb20}
\mathbb{P}\left(\mbox{no $\left(n_{1},\cdots,n_{\ell}\right)$ exists in $[t-k(J),t]$}\right)&=&\mathbb{P}\left(\tau_{t}^{J}>t\right)\nonumber \\ & = & \sum_{z\in\mathcal{Z}}\left[\mathbb{P}\left(z(t-k(J)-1)=z\right)\right.\nonumber \\ & &
\nonumber
\left.\mathbb{P}\left(\tau_{t}^{J}>t~|~z(t-k(J)-1)=z\right)\right]\\
\label{sb21}
&=&\sum_{z\in\mathcal{Z}}\mathbb{P}\left(z(t-k(J)-1)=z\right)\mathbb{P}_{z}\left(\tau_{0}>k(J)+1\right)
\end{eqnarray}
Since the above development holds for all $t\geq t_{0}$ for some sufficiently large $t_{0}$, we conclude from eqn.~(\ref{sb5})
\begin{equation}
\label{sb25}
\sup_{t\geq t_{0}}\mathbb{P}\left(\left\|\widetilde{P}(t)\right\|>J\right)\leq \sum_{z\in\mathcal{Z}}\mathbb{P}\left(z(t-k(J)-1)=z\right)\mathbb{P}_{z}\left(\tau_{0}>k(J)+1\right)
\end{equation}
The recurrence (in fact positive recurrence) of the finite state Markov chain $\{z(t)\}$ and the fact that $k(J)\rightarrow\infty$ as $J\rightarrow\infty$ imply, for all $z\in\mathcal{Z}$,
\begin{equation}
\label{sb22}\lim_{J\rightarrow\infty}\mathbb{P}_{z}\left(\tau_{0}>k(J)+1\right)=0
\end{equation}
Since $\mathcal{Z}$ is finite, letting $J\rightarrow\infty$ in eqn.~(\ref{sb25}) leads to
\begin{equation}
\label{sb23}
\lim_{J\rightarrow\infty}\sup_{t\geq t_{0}}\mathbb{P}\left(\left\|\widetilde{P}(t)\right\|>J\right)=0
\end{equation}
by the dominated convergence theorem and the Lemma follows.
}

\bibliographystyle{IEEEtran}
\bibliography{IEEEabrv,CentralBib}

\end{document}